\documentclass[11pt]{article}
\usepackage{fullpage,amssymb,amsmath,amsthm,graphicx,natbib,hyperref}

\newcommand{\etr}{{\rm etr}}

\newcommand{\Exp}[1]{{\rm E}[ \ensuremath{ #1 } ]  }
\newcommand{\Var}[1]{{\rm Var}[ \ensuremath{ #1 } ]  }
\newcommand{\Cov}[1]{{\rm Cov}[ \ensuremath{ #1 } ]  }

\newcommand{\bl}[1]{{\mathbf #1}}
\newcommand{\bs}[1]{{\boldsymbol #1}}

\newtheorem{prop}{Proposition}

\title{Equivariant and scale-free Tucker decomposition models}
\author{Peter David Hoff \\
Departments of Statistics and Biostatistics \\
University of Washington}
\date{\today}

\begin{document}

\maketitle

\begin{abstract}
Analyses of array-valued datasets
often involve
reduced-rank array approximations, typically obtained
via least-squares or truncations of 
array decompositions. 
However, least-squares approximations tend to  be noisy
in high-dimensional settings, and may not be appropriate
for arrays that include discrete or ordinal
measurements.
This article
develops methodology
to obtain low-rank model-based representations of
continuous, discrete and ordinal data arrays.
The model is based on a
parameterization of the  mean array
as a multilinear product of a reduced-rank core array and a set of
index-specific orthogonal eigenvector matrices.
It is shown how
orthogonally equivariant parameter estimates can be obtained
from Bayesian procedures under invariant prior distributions.
Additionally,
priors on the core array are developed that act as regularizers,  leading to
improved inference over the standard least-squares estimator, 
and  providing
robustness to
misspecification of the array rank. 
This model-based approach  is extended to accommodate discrete or ordinal
data arrays using a
semiparametric transformation model.
The resulting low-rank representation is scale-free, in the sense that it is invariant to monotonic transformations of the data array.
In an example analysis of a multivariate discrete 
network dataset, 
this scale-free approach provides a more
complete description of data patterns. 

\smallskip


\smallskip

\end{abstract}

\section{Introduction}
Many datasets are naturally represented
as multiway arrays, often referred to as tensors. 
For example, 
data gathered under all combinations of levels of three 
conditions can be 
expressed as
a three-way array 
$\bl Y  =\{ y_{i,j,k} : i\in \{1,\ldots, n_1\} , j \in \{1,\ldots, n_2\} , 
  k \in \{1,\ldots, n_3\}\}$. 
The index sets are referred to as the modes of the array, 
and
an array with $K$ modes is typically referred to as a $K$-way array. 
Such array-valued datasets are common in several disciplines, 
including 
chemometrics, signal processing and psychometrics. 
Another class of array-valued data includes multivariate 
relational networks, which consist of several types of relational measurements
between pairs of nodes.  
Such a dataset may be represented as a three-way array  
$\bl Y \in \mathbb R^{n\times n\times p}$,   
where $n$ is the number of nodes, $p$ is the number of relation types, 
and the entries of $\bl Y$ are such that $y_{i,j,k}$ is the 
value of the $k$th relation type from node $i$ to $j$. 
For example, 
$y_{i,j,1} $ may give the number of emails sent from person $i$ to
person $j$ 
and $y_{i,j,2}$ may encode an evaluation
of  $i$'s  friendship to $j$ 
measured on an ordinal scale.
In this case, the three modes of the array correspond to 
the initiator of the relation, the target of the relation and 
the relation type, respectively. 

A popular method for describing heterogeneity in 
array-valued datasets is with array decompositions.  
One category of decompositions are the ``Tucker decompositions''  
\citep{tucker_1964, tucker_1966,kolda_bader_2009}, which 
express a $K$-way data array 
$\bl Y$  as 
$\bl Y = \bl S \times  \{ \bl U_1,\ldots, \bl U_K\}$, 
where 
$\bl S$ is a $K$-way core array, 
``$\times$'' is a  a  multilinear operator known as  the Tucker product and 
$\{\bl U_1,\ldots, \bl U_K\}$ is a collection of 
mode-specific factor matrices. 
\citet{delathauwer_demoor_vandewalle_2000} 
study a particular type of Tucker decomposition  in which 
the $\bl U_k$'s are orthogonal, and argue 
that this ``higher-order'' singular value decomposition (HOSVD) is 
a natural extension of the matrix SVD
to  arrays, with the core array $\bl S$ playing a role 
analogous to that of the 
singular values of a matrix. 
Data analysis based on this decomposition often proceeds by 
obtaining a low-rank representation of $\bl Y$ either 
via truncation of the core array or with a 
least-squares approximation, and then 
using its  mode-specific singular vectors
to describe the heterogeneity 
in the entries of $\bl Y$ along each of its $K$ modes.

While providing a relatively simple approach to exploratory data-analysis, 
least-squares methods 
may be limited in terms of their performance and  applicability. 
For example, 
least-squares methods tend to 
be noisy in multiparameter estimation problems,
leading many researchers to favor 
regularized procedures instead. 
Recent work on the analysis 
of matrix-valued datasets  
indicates that soft-thresholding the singular values of a
data  matrix can lead to improved estimation of its 
mean matrix as compared to a least-squares approach
\citep{
mazumder_hastie_tibshirani_2010, 
cai_candes_shen_2010,
josse_sardy_2013}. 
Penalized 
approaches have also been studied in the context of array-valued data:
Recent theoretical work has focused on array completion problems,
in which the  task is to recover a reduced-rank array based on
random linear combinations
of its elements \citep{liu_musialski_przemyslaw_wonka_ye_2009, mu_huang_wright_goldfarb_2013}. The algorithms studied typically involve 
finding the minimum rank among 
arrays  that match the data at the observed entries.
Variants of these procedures include finding arrays that minimize
different criteria while still matching the observed data,
or by minimizing a residual sum of squares subject to a penalty
on the fitted array \citep{tomioka_suzuki_hayashi_kashima_2011}.

However,  
such approximations of the raw data
may be inappropriate when the data are
binary, ordinal or otherwise non-normally distributed.
For example, Section 5 of this article
considers an analysis of skewed, discrete
multivariate relational data. These data,
obtained from the GDELT project
(\citet{leetaru_schrodt_2013}, \href{http://gdelt.utdallas.edu/}{gdelt.utdallas.edu}),
consist of  weekly summaries of
20 different types of actions
between the 30 most active countries in the GDELT database in
2012. These data can be represented as a $30\times 30 \times 52 \times 20$
four-way array $\bl Y$, with entries
$\{ y_{i,j,k,t} : 1 \leq i , j \leq 30, i\neq j, 1\leq k \leq 20, 
  1\leq t\leq 52\}$,
where $y_{i,j,k,t}$ is the number of days in
week $t$ in which
country $i$ took action $k$ with country $j$ as the target.
A least-squares approximation to these data is
problematic for several reasons, one of which is that
such an approximation
predominantly
represents the  small number of large entries of the array, and is therefore
unrepresentative of ``most'' of the data. 

As an alternative to least-squares procedures,  
this article develops a model-based version of 
a penalized Tucker decomposition, and an extension that can accommodate the 
analysis of discrete, ordinal or otherwise non-normal data. 
The approach is Bayesian, 
in that the penalty term can be viewed as a prior distribution on the 
unknown parameters, and estimates can be obtained via Markov chain 
Monte Carlo methods. 
This Bayesian model-based approach is similar to that of 
\citet{chu_ghahramani_2009}, who present a Tucker decomposition 
model and prior in which the core array $\bl S$ and factor matrices 
$\{ \bl U_1,\ldots, \bl U_K\}$ all have 
i.i.d.\ standard normal entries. 
Unlike their approach, 
this article parameterizes the model so that the factor matrices 
are orthogonal, 
as in the HOSVD of 
\citet{delathauwer_demoor_vandewalle_2000}. 
This parameterization facilitates construction of a 
class of prior distributions  for which 
posterior inference is both scale-equivariant and 
orthogonally-equivariant. 
Additional 
identifiability considerations lead to a particular 
form for a prior distribution over the core array $\bl S$. 
This prior 
allows for mode-specific penalization of the 
singular values, and also has an interpretation 
as a version of normal factor analysis for array-valued data. 

The
work presented here is related to some recently developed
statistical models that make use of the multilinear Tucker product.
The core array $\bl S$ is penalized using 
a class of array normal distributions, generated by the multilinear Tucker
product
\citep{hoff_2011b}. 
\citet{xu_yan_qi_2012} develop a prior over the array
normal model in which the mode-specific covariance matrices are
functions of a potentially infinite set of latent features.
In a similar vein, 
\citet{fosdick_hoff_2012}  
develop a version of factor analysis based on the array normal model. 
The Tucker product has also been used
to construct priors
in applications where it is the parameters in the model
that are arrays:
\citet{bhattacharya_dunson_2012} use a Tucker product to develop
a prior over probability distributions for multivariate
categorical data, and
\citet{volfovsky_hoff_2012} use a collection
of connected array normal distributions
as a prior over parameter arrays in ANOVA decompositions.
Regarding penalization,
\citet{allen_2012} has proposed
a sparsity penalty on the factor matrices of a Tucker decomposition,
thereby
encouraging zeros in their entries.
While appropriate in some applications, procedures based on
such a sparsity penalty will  not be orthogonally equivariant.
In contrast, the uniform priors
on the factor matrices  used in this article
lead to orthogonally equivariant estimates,
and penalization is focused on the core array in order to 
encourage low-rank approximations
to the data.

An outline of this paper is as follows: 
The next section provides a brief review of 
array rank and Tucker decompositions. 
In Section 3 a 
parameterization of the Tucker decomposition model is presented, 
along with a class of prior distributions 
that allow for equivariant estimation of the model 
parameters.   Section 4 develops
 a  subclass of priors
that allows
for mode-specific penalization of the singular values. 
In a simulation study, 
this prior distribution is shown to perform as well as an 
``oracle'' prior when no mode-specific penalization is warranted, 
and greatly outperforms such a prior when the rank of the model 
is misspecified. 
This methodology  is extended in Section 5 to accommodate 
discrete, ordinal and non-normal data 
via a semiparametric transformation model, 
allowing for scale-free reduced-rank representations of array data 
of diverse types. 
This extension  is illustrated with an analysis of 
discrete multivariate international relations data.
A discussion follows in Section 6.


\section{Review of array rank and Tucker decompositions} 
Recall that the rank of a matrix $\bl M\in \mathbb R^{n_1\times n_2}$ 
is equal to the dimension of the linear space 
spanned by the columns (or rows) of $\bl M$. 
Now suppose $\bl M\in \mathbb R^{n_1\times n_2\times n_3}$  
is a three-way array, with elements 
$\{ m_{i,j,k}  : 1\leq i \leq n_1 ,1\leq j \leq n_2,1\leq k \leq n_3 \}$.
The notion of array rank considered by 
\citet{tucker_1964}, \citet{delathauwer_demoor_vandewalle_2000} and others 
is defined by  
the ranks of various reshapings of $\bl M$
into matrices, called {\it matricizations}. 
For example,
the mode-1 matricization $\bl M_{(1)}$ of $\bl M$ is the
$n_1 \times (n_2 n_3)$ matrix having column vectors of the form
$\bl m_{j,k} = (m_{1,j,k},\ldots, m_{n_1,j,k})^T$, that is, 
elements of $\bl M$ with varying values of the first 
index and fixed values of the second and third indices. 
Heterogeneity in the values of $\bl M$ ascribable to 
heterogeneity in the first index set can be described in terms 
of the linear space spanned by the columns of 
$\bl M_{(1)}$. The dimension $r_1$ of this linear space (which 
is equal to the rank of $\bl M_{(1)}$) is called 
the {\it mode-1 rank} of $\bl M$. 
The mode-2 and mode-3 matricizations of $\bl M$ can be formed similarly, 
and their ranks provide the mode-2 rank $r_2$ and mode-3 rank $r_3$, 
respectively. The array rank of $\bl M$ is 
the vector $\bl r= (r_1,r_2,r_3)$, and is sometimes referred to as the multilinear rank.  Unlike the row and column ranks of a matrix, 
the ranks corresponding 
to the different modes of an array are not generally equal. 

Any matrix $\bl M \in \mathbb R^{n_1\times n_2}$ can
be expressed in terms of its SVD
$\bl M = \bl U_1 \bl S \bl U_2^T$
where  $\bl S =  \text{diag}(s_1,\ldots, s_r)$,
 $\bl U_1\in \mathcal V_{r,n_1}$, $\bl U_2\in \mathcal V_{r,n_2}$ 
and  $r\leq n_1\wedge n_2$ is the rank of $\bl M$.
Here, $\mathcal V_{r,n}$ is the space of $n\times r$ matrices with 
orthonormal columns, known as the Stiefel manifold. 
As shown by
\citet{delathauwer_demoor_vandewalle_2000},
an analogous representation
holds for any array. The analogy is most easily seen 
via vectorization: 
The SVD of a matrix $\bl M$ yields a representation of 
$\bl m=\text{vec}(\bl M)$ 
as
$\bl m = (\bl U_2 \otimes \bl U_1 )\, \bl s$,
where $\bl s= \text{vec}(\bl S)$ and 
``$\otimes$'' is the Kronecker product. 
Similarly, 
every  $K$-way array $\bl M$
of dimension $n_1\times \cdots \times n_K$ and rank 
$\bl r=( r_1,\ldots, r_K)$
can be expressed as 
\begin{equation}
 \bl m = ( \bl U_K \otimes \cdots \otimes \bl U_1 ) \, \bl s,  
\label{eq:tkprod}
\end{equation}
where $\bl m$ is the vectorization of $\bl M$, $\bl U_k \in \mathcal V_{r_k,n_k}$
for $k\in \{1,\ldots, K\}$ and $\bl s$ is the vectorization of an 
$r_1\times \cdots \times r_K$ array $\bl S$ known as the ``core array.'' 
This representation is 
often referred to as the higher-order SVD (HOSVD). 
More generally, any representation of $\bl m$
of the form (\ref{eq:tkprod}), without
$\bl U_1,\ldots, \bl U_K$
necessarily being orthogonal, is known as a ``Tucker decomposition.''

An equivalent representation of $\bl M$ that retains its array structure is 
obtained using the so-called ``Tucker product'' \citep{tucker_1964}
of the core array $\bl S$ with the list of factor
matrices $\bl U_1,\ldots, \bl U_K$. 
This representation expresses $\bl M$ as
\begin{equation} 
\bl M = \bl S \times  \{ \bl U_1,\ldots, \bl U_K\}, 
\label{eq:tprod}
\end{equation}
where the Tucker product ``$\times$''
is defined by the equivalence 
between Equations \ref{eq:tkprod} and \ref{eq:tprod}. More 
generally,  For $\bl A\in \mathbb R^{n_1\times \cdots \times n_K}$, 
 $\bl B\in \mathbb R^{r_1\times \cdots \times r_K}$ 
and $\bl C_k \in \mathbb R^{n_k \times r_k}$, $k=1,\ldots, K$,  
$\bl A = \bl B \times  \{ \bl C_1,\ldots, \bl C_K\}$ means that 
  $\text{vec}(\bl A) = ( \bl C_K \otimes \cdots \otimes \bl C_1 ) \, \text{vec}(\bl B) $. 

For the calculations that follow it will be useful to re-express a
Tucker decomposition
of $\bl M$ in terms of its matricizations. If $\bl M$ can 
be expressed as in (\ref{eq:tkprod}) or   (\ref{eq:tprod}), then it also 
follows that for each $k\in \{1,\ldots, K\}$, 
\begin{equation}
\bl M_{(k)}  =  \bl U_{k} \, \bl S_{(k)} \, (\bl U_K \otimes \cdots 
   \bl U_{k+1} \otimes  \bl U_{k-1} \otimes \cdots 
\otimes \bl U_1 )^T  \\
 \equiv \bl U_{k} \, \bl S_{(k)} \, \bl U_{-k}^T, 
\end{equation}
where $\bl M_{(k)}$ and $\bl S_{(k)}$ are the mode-$k$ 
matricizations of $\bl M$ and $\bl S$ respectively.

\section{A model-based Tucker decomposition for arrays}
A commonly used model of low-dimensional structure for a matrix-valued dataset
$\bl Y \in \mathbb R^{n_1\times n_2}$  is  that $\bl Y$
is equal to some mean matrix $\bl M$ of rank $r<n_1\wedge n_2$, 
plus an error matrix $\sigma \bl E$ 
having 
i.i.d.\  mean-zero entries with variance $\sigma^2$. 
Let $\bl M= \bl U_1 \bl D \bl U_2^T$ be the  SVD
of $\bl M$ and $\bl S= \bl D/\sigma$ 
be the singular values scaled by the error standard deviation $\sigma$. 
This model can be parameterized as
$\bl Y = \sigma  \bl U_1 \bl S \bl U_2^T + \sigma \bl E$, or alternatively 
in vector form 
as
\[  \bl y = \sigma (\bl U_2 \otimes \bl U_1) \, \bl s   + \sigma \bl e, \]
where $\bl y$, $\bl s$ and $\bl e$ are the vectorizations of 
$\bl Y$, $\bl S$ and $\bl E$ respectively. 

Now consider an analogous model for an array 
$\bl Y \in \mathbb R^{n_1\times \cdots\times n_K}$.  As in the matrix case, 
the model is $\bl Y = \bl M + \sigma \bl E$, where $\bl M$ is 
an array with array rank $\bl r$ and 
$\bl E$ is a mean-zero error array.  
Equation \ref{eq:tkprod} says that this model 
can be expressed as 
\begin{equation} 
  \bl y = \sigma ( \bl U_K \otimes \cdots \otimes \bl U_1 ) \, \bl s + 
 \sigma \bl e  
\label{tsvdm}
\end{equation}
where 
$\bl s \in \mathbb R^{r_1\cdots r_K}$ and 
 $\bl U_k \in \mathcal V_{r_k,n_k}$ for each $k=1,\ldots, K$. 
An equivalent representation in terms of the Tucker product is that 
\[
  \bl Y = \sigma \bl S \times \{ \bl U_1,\ldots, \bl U_K\} + \sigma \bl E. \]
This section discusses estimation of the unknown parameters
$(\sigma, \bl U , \bl S)$
in this Tucker decomposition model (TDM) 
when the error $\bl E$ is assumed to consist of  i.i.d.\ standard 
normal random variables. 
Results on optimal equivariant estimation in the case that 
$\bl S$ is known are used to motivate certain 
priors for equivariant 
Bayesian inference in the more realistic case that $\bl S$ 
is unknown. 
It is shown that 
posterior inference under such prior distributions can be made 
with a relatively straightforward 
Markov chain Monte Carlo (MCMC) algorithm, based on 
Gibbs sampling.

\subsection{Equivariant estimation}
First consider the (unrealistic) case that the core array $\bl S$  
is known. 
Letting $n= n_1\cdots n_K$, $r= r_1\cdots r_K$ and  
$\mathcal U = \{ \bl U : \bl U =  \bl U_K \otimes \cdots \otimes \bl U_1, 
 \bl U_k \in \mathcal V_{r_k,n_k} \}$, 
 the normal TDM can be expressed as
\begin{equation}
 \bl y  = \sigma \bl U \bl s + \sigma \bl e  \label{eq:rtsvdmod}  \ , \  
 \bl e   \sim  N_n(  \bs 0 , \bl I )  \ , \
(\sigma , \bl U )   \in   \mathbb R^+  
   \times \mathcal U .  
\end{equation}
Let
 $\mathcal W = \{  
  \bl W : \bl W =  \bl W_K \otimes \cdots \otimes \bl W_1, 
 \bl W_k \in \mathcal O_{n_k}  \}$
be the space of Kronecker products of orthogonal matrices, and 
note that $\bl W \bl U \in \mathcal U$ for all $\bl W \in \mathcal W $
and $\bl U \in \mathcal U$.
It follows that the model (\ref{eq:rtsvdmod}) is invariant under the
group of transformations  on $\mathcal Y$
given by $\mathcal G = \{ g: \bl y \rightarrow  
   a \bl W \bl y, a >0, \bl W \in \mathcal W   \}$, which
 induces a group
$\bar{ \mathcal G}$ on the parameter space given by
$\bar{ \mathcal G} = \{  \bar g : ( \sigma, \bl U ) \rightarrow
 (a \sigma, \bl W \bl U  )   \}$.
This motivates the use of equivariant estimators of $\sigma$
and $\bl U$. For example, it is natural to prefer estimators such 
that  $\hat \sigma ( a \bl W \bl y) = a \hat \sigma(\bl y)$,
so that the scale changes to the data result in the same change to
the estimate of the scale
parameter $\sigma$.   Similarly, one may prefer
estimators of $\bl U$ such that 
$\hat {\bl U}( a \bl W \bl y ) = \bl W  \hat { \bl U }(\bl y)$
and estimators of $\bl m = \sigma \bl U \bl s$ such that 
$\hat {\bl m}(a\bl W \bl y) = a \bl W \hat {\bl m}(\bl y)$.

As with many invariant statistical models,
risk-optimal equivariant decision rules 
can be obtained as  Bayes 
rules under a prior distribution derived from the group:
\begin{prop} 
Let $\theta = (\sigma, \bl U)$ and $\Theta = \mathbb R^+ \times 
  \mathcal U$. Under any invariant loss function  
$L(d,\theta)$  the minimum risk equivariant decision rule 
$\delta(\bl y)$ is given for each $\bl y$ by 
the minimizer in $d$ of 
\[   \int L(d,\theta) p(\bl y|\theta) \pi_I(d \theta),  \] 
where for measurable sets $A \subset \mathbb R^+$ and 
  $B\subset \mathcal U$, 
$\pi_I(A\times B ) = \pi_\sigma(A) \times \pi_U(B)$, with 
$\pi_\sigma(A) = \int_A  \sigma^{-1}  \, d\sigma $ and
$\pi_U$ corresponding to the (proper) probability distribution of 
   $\bl U_K\otimes \cdots \otimes \bl U_1$ when
each  $\bl U_k$ is uniformly distributed on $\mathcal V_{r_k,n_k}$. 
\label{prop:invdr}
\end{prop}
This result is an application of more general results from
invariant decision theory 
(a proof is in the Appendix).
To put the result more simply, optimal equivariant decision rules
can be obtained from the posterior distribution of  
$(\sigma, \bl U)$  under an improper
prior for $\sigma$ with density $1/\sigma$ 
and independent uniform priors for $\bl U_1,\ldots, \bl U_K$. 
In what follows, $\pi_\sigma$ and $\pi_U$  will refer to 
either these measures or their densities, depending on the context.

Unfortunately, uniformly optimal equivariant decision rules no longer 
exist 
under this group
when the core array $\bl s$ is unknown, 
as the best equivariant estimator will depend on $\bl s$. 
This article focuses
attention on 
Bayesian inference for $(\sigma, \bl U, \bl s)$  using  
prior distributions with densities of the form 
$\pi(\sigma , \bl U , \bl s) = \pi_\sigma(\sigma)  \pi_U(\bl U) \pi_s(\bl s)$, 
where $\pi_s(\bl s)$ is a proper probability density. 
Although not corresponding to a proper joint prior distribution
(because of the improper prior on $\sigma$), 
such densities can be used to construct proper 
posterior distributions that provide
estimates of functions of $(\sigma,\bl U,\bl s)$
that are equivariant with respect to $\mathcal G$ and 
$\bar {\mathcal G} = \{ \bar g: (\sigma,\bl U, \bl s) \rightarrow 
   (a\sigma, \bl W \bl U, \bl s)\}$. 
Addressing the propriety of such a  posterior first,  
for each $\bl y \in \mathbb R^n$ 
define a function $f(\sigma, \bl U , \bl s: \bl y)$  so that 
\[  f( \sigma, \bl U , \bl s: \bl y) \propto
     p(\bl y  |\sigma, \bl  U , \bl s) \times  \pi(\sigma, \bl U, \bl s) ,\]
where $p(\bl y|\sigma, \bl U, \bl s) $ is the normal sampling 
density of $\bl y$, having mean $\sigma \bl U \bl s$ and 
variance $\sigma^2 \bl I$. 
If $f$ is integrable in $(\sigma, \bl U, \bl s)$ for the observed value 
of $\bl y$,  
a ``posterior'' probability distribution can be defined via the 
density
\begin{equation} \pi( \sigma, \bl U , \bl s | \bl y)   = 
     \frac{f(\sigma, \bl U , \bl s :\bl y) }
{ \int f(\sigma, \bl U , \bl s :\bl y) \ d\sigma  d\bl U  d\bl s } \, .  
\label{eq:gbayespost}
\end{equation}
That $f$ is generally integrable can be seen by first 
integrating with respect to $\sigma$:
\begin{align*} 
\int_0^\infty f(\sigma, \bl U, \bl s :\bl y )  \, d\sigma &= 
 \pi_U(\bl U) \pi_s(\bl s) \int_0^\infty  
    p(\bl y| \sigma, \bl U , \bl s) \pi_\sigma(\sigma) \, d\sigma  \\
&= 
 \pi_U(\bl U) \pi_s(\bl s) \int_0^\infty (2\pi)^{-n/2} \sigma^{-n-1} \exp( -\sigma^{-2} ||\bl y -\bl U \bl s||^2/2 ) \, d\sigma   \\
&= \pi_U(\bl U) \pi_s(\bl s) \times   [\tfrac{1}{2} \pi^{-n/2} \Gamma(n/2 )]  
 \times     ||\bl y-\bl U\bl s||^{-n} \, . 
\end{align*}
Now
$||\bl y - \bl U \bl s|| \geq   ||\bl y - \hat {\bl m}|| 
 $, where $\hat {\bl m}$ is the least squares
estimate of $\bl m$.  Since $\hat {\bl m}$ is of  reduced rank, 
 $||\bl y - \hat {\bl m}||>0$ unless the array rank 
of $\bl y$ is less than or equal to that of 
the fitted rank. 
Presuming this is not the case,
it follows that  $||\bl y-\bl U\bl s||^{-n} $ is bounded above
by  $||\bl y - \hat {\bl m} ||^{-n}$. Since the
 priors for $\bl U$ and $\bl s$ are proper, the integral of
   $||\bl y -  {\bl U} {\bl s}||^{-n}$ with respect to
 $\pi_U(\bl U)$ and $\pi_s(\bl s)$ is finite. 
Therefore, $f( \sigma, \bl U , \bl s: \bl y) $ is integrable 
and (\ref{eq:gbayespost}) is a proper probability density. 

As stated above, the decision rules  obtained from such a posterior 
are not globally risk optimal among equivariant rules, 
as optimal rules for $(\sigma, \bl U)$ depend on the unknown 
value of $\bl s$. However, such posterior distributions still 
provide equivariant inference in the following sense: 
\begin{prop}
Let the prior for $\theta=(\sigma , \bl U, \bl s)$ be such that 
the marginal prior for $(\sigma, \bl U)$ is the invariant
prior $\pi_I$ and 
$\bl s$ is independent of $(\sigma, \bl U)$.
Then for any 
$a>0$, $\bl W\in \mathcal W$ and functions 
$ g:\bl y \rightarrow  a \bl W \bl y$ 
and  $\bar g: (\sigma, \bl U , \bl s ) \rightarrow  
     (a \sigma, \bl W \bl U , \bl s) $, 
\[  \Pr( \theta \in A |\bl y ) = \Pr(\theta\in \bar g A |g\bl y)    \]
\label{prop:eqinf}
for all measurable subsets $A$ of $\mathbb R^+ \times \mathcal U \times 
  \mathbb R^r$. 
\end{prop}
A proof is in the Appendix. 
The result says that, using such a prior, 
the  belief that the correct $\theta$-value is in $A$ having 
observed $\bl y$ is the same  
as the  belief that the correct $\theta$-value is in $\bar gA$ having  
observed $g\bl y$. 

\subsection{Posterior approximation via the Gibbs sampler}
The results in the previous subsection hold 
as long as $\bl s$ is \emph{a priori} independent of $\sigma$ and $\bl U$
and the prior for $\bl s$ is proper.  The remainder of the 
article focuses 
attention 
on normal priors for $\bl s$, so that 
the joint prior distribution of $(\sigma, \bl U, \bl s)$ 
has  
a density of the form 
$\pi(\sigma, \bl U , \bl s)  =  \pi_I(\sigma , \bl U) \times 
    \pi_s(\bl s)$,
where $\pi_I$ is density of the invariant prior discussed previously  and
$\pi_s$ is  a zero-mean multivariate normal prior  with covariance
matrix $\bs \Psi$.
Not only are such priors for $\bl s$ computationally convenient, 
but they lead to an interpretation of the model as a multiway extension to a 
normal factor analysis model, as will be discussed in the next section. 

Posterior inference  under such a prior can be made via a reasonably 
straightforward Gibbs sampling algorithm that 
approximates the posterior distribution of 
$( \sigma^2 ,   \bl U ,  \bl s )$ 
given $\bl y$. 
The algorithm proceeds by iteratively updating the values of 
these parameters as follows: 
\begin{enumerate}
\item Simulate $(\sigma^2,\bl s) $ from $\pi(\sigma^2, \bl s| \bl y, \bl U  )$ as follows: 
\begin{enumerate}
\item simulate $\sigma^2$ from $\pi(\sigma^2| \bl y, \bl U  )$, an inverse-Gamma distribution; 
\item simulate $\bl s $ from $\pi(\bl s | \bl y, \bl U ,\sigma^2)$, 
a multivariate normal distribution. 
\end{enumerate}
\item For $k\in \{1, \ldots, K\}$, 
 simulate 
   $\bl U_k $ from $\pi(\bl U_k| \bl y, \bl s , \{ \bl U_j:j\neq k\} ,\sigma^2 )$, a von Mises-Fisher distribution on $\mathcal V_{r_k,n_k}$.  
\end{enumerate} 
Repeated iteration of the above procedure generates a Markov 
chain whose stationary distribution is the posterior distribution of
$( \sigma^2, \bl U, \bl s )$
given $\bl y$.

\paragraph{Full conditional distribution of $(\sigma^2,\bl s)$:}  
Recall that the model for $\bl y$ is 
$\bl y= \sigma \bl U \bl s + \sigma \bl e, \ \bl e \sim N_n(\bs 0,\bl I)$, 
where $n=\prod n_k$. 
The normal prior $\bl s\sim N_{r}(\bl 0, \bs\Psi)$ 
implies that, unconditionally on $\bl s$, $\bl y$ is multivariate 
normal with mean $\bl 0$ and covariance matrix 
\begin{align*} 
\Exp{ \bl y \bl y^T | \bl U, \sigma } & = 
  \sigma^2   \Exp{  \bl U \bl s \bl s^T  \bl U^T  + \bl e\bl e^T  + 
    2 \bl U \bl s \bl e^T}  \\
  &= \sigma^2 \left ( \bl U \bs \Psi \bl U^T + \bl I \right). 
\end{align*}
Based on this result, standard calculations show that 
the conditional distribution of $\sigma^2$ 
used in step 1  of the above algorithm is an  
inverse-gamma distribution:
\begin{align*} 
1/\sigma^2 & \sim \text{gamma}(  n /2 , 
        \bl y^T (  \bl U \bs \Psi \bl U^T  + \bl I )^{-1}  \bl y/2   ). 
\end{align*} 
Now given $\sigma$ and $\bl U$, 
the model can be expressed as 
$\bl y/\sigma = \bl U \bl s + \bl e$
where the entries of $\bl e$ are i.i.d.\ standard normal random variables.
This has the same form as
a regression model with $\bl s$ playing the role of
the vector of  unknown regression coefficients.
Combining this ``regression likelihood'' with the normal prior
 $\bl s \sim N_{r} (\bl 0 , 
   \bs \Psi )$
gives a normal full conditional distribution for $\bl s$ with
mean and variance given as follows:
\begin{align*}
\Var{ \bl s| \bl y, \bl U, \sigma^2, \bs \Psi }  &= \tilde {\bs \Psi} = 
  ( \bs \Psi^{-1}  +  \bl I )^{-1}  \\
\Exp{\bl s| \bl y, \bl U, \sigma^2, \bs \Psi }  &= 
    \tilde {\bs \Psi }  \bl U^T \bl y/\sigma. 
\end{align*}
The next section discusses
specification and 
 estimation 
of $\bs \Psi$, 
and its relationship to 
the mode-specific singular values of 
the  mean array
$\bl M$.

\paragraph{Full conditional distribution of $\bl U$:}
Let $\bl Y_{(1)}$, $\bl S_{(1)}$ and  $\bl E_{(1)}$ be the
mode-1 matricizations of the arrays $\bl Y$, $\bl S$ and $\bl E$
respectively.  The model can then be written as
$\bl Y_{(1)}/\sigma = \bl U_1  \bl S_{(1)} \bl U_{-1}^T  + \bl E_{(1)}$
where $\bl U_{-1} = ( \bl U_K \otimes  \cdots\otimes  \bl U_2)$ 
and the elements of $\bl E_{(1)}$ are i.i.d.\ standard normal random 
variables. 
Since the prior for $\bl U_1$ is the uniform distribution on $\mathcal V_{r_1,m_1}$, its full conditional distribution is proportional to
the density of $\bl Y_{(1)}$:
\begin{align*}
\pi(\bl U_1 | \ldots)   \propto_{\bl U_1}   p(\bl Y_{(1)} |\bl S, \bl U , \sigma^2_e)  &  \propto_{\bl U_1}  \exp( -\tfrac{1}{2}  ||\bl Y_{(1)}/\sigma - \bl U_1\bl S_{(1)} \bl U_{-1}^T    ||^2   )  \\
&\propto_{\bl U_1}  \etr( \bl U_1^T \bl Y_{(1)} \bl U_{-1} \bl S_{(1)}^T  )/\sigma ) 
\equiv  \etr( \bl U_1^T \bl H  ) 
\end{align*}
where $\bl H = \bl Y_{(1)} \bl U_{-1} \bl S_{(1)}^T /\sigma$.
This is proportional to the matrix-variate
von Mises-Fisher distribution
vMF$(\bl H)$
on $\mathcal V_{r_1,m_1}$.  An algorithm for direct simulation from
 vMF$(\bl H)$ is described in \citet{hoff_2009a}.
The full conditional distributions of $\bl U_2,\ldots, \bl U_K$ 
can be derived analogously.


\section{Estimation of \texorpdfstring{$\bs \Psi$}{Psi} }
The covariance matrix $\bl \Psi$ of the core array $\bl S$ 
can be viewed as a description of the 
scale of $\bl M$ relative to the scale $\sigma$ of the error, 
or alternatively, as a penalty on the magnitude of $\bl S$ 
that serves to provide a regularized estimator of the mean 
array $\bl M = \sigma \bl S \times \bl U$. In practice, 
an appropriate value of $\bl \Psi$ may not be known 
in advance, and therefore must be estimated from the data. 
This section discusses estimation of $\bl \Psi$ in the context of 
two models for $\bl S$. The first of these is 
simply that $\text{vec}(\bl S) = \bl s \sim N_r(\bl 0, \tau^2 \bl I)$, 
where $\tau^2$ is a 
scale parameter to be estimated. 
In a simulation study, it is shown  that
this  model provides better estimates of $\bl M$ than those obtained by 
minimizing the residual sum of squares. 
However, this simple covariance model shrinks all values of $\bl S$ 
equally, and does not recognize the array structure of $\bl S$. As 
an alternative to this homoscedastic 
i.i.d.\ model, a heteroscedastic 
separable variance
model is developed,  of the form 
$\Cov{\bl s} = \tau^2 \bs\Lambda_K \otimes \cdots \otimes  
  \bs \Lambda_1$, 
where each $\bs \Lambda_k $ is a diagonal matrix with positive 
entries that sum to 1.
Such a model allows for separate penalization of the 
mode-specific eigenvalues of the array $\bl M$. 
Such penalization is useful when it is feared that 
the fitted rank $\bl r$ is larger than the actual rank of the 
mean array for some of the modes.  
In such cases, it is desirable to have  a procedure that can shrink the 
estimate of $\bl M$ towards arrays with lower 
mode-specific ranks. 
This section first derives
 this heteroscedastic model 
and provides some 
interpretation of the parameters,  and then illustrates in a simulation 
study 
how  estimators based on this model 
can shrink towards low-rank solutions when the fitted rank 
is too large. 

\subsection{Derivation and interpretation of the heteroscedastic model}
Even if $\bl s$ were observed, unrestricted estimation of $\bl \Psi$ 
based on the model $\bl s\sim N_r(\bl 0, \bl \Psi)$ 
would be problematic, as $\bl s$ corresponds to  only 
a single realization from the $N_r(\bl 0, \bl \Psi)$ distribution. 
Instead, consider  first estimation of $\bs \Psi$ restricted 
to the class of separable covariance matrices, so 
that $\bs \Psi  = \bs \Psi_K \otimes \cdots \otimes \bs \Psi_1$, where 
each $\bs \Psi_k$ is an $r_k\times r_k$ positive definite matrix.  
Now recall that marginally over $\bl s$, the distribution for 
$\bl y=\text{vec}(\bl Y)$ is a mean-zero $n$-variate normal distribution 
with covariance matrix proportional to $\bl U \bs \Psi \bl U^T +\bl I$. As $\bl U$ and $\bs \Psi$ 
are both separable,  it follows that
\[  \Cov{\bl y| \bl \sigma , \bl U , \bs\Psi }/\sigma^2 =
    \bl U\bs\Psi \bl U^T + \bl I   = 
    (\bl U_K\bs\Psi_K \bl U_K^T \otimes \cdots \otimes 
                 \bl U_1\bs\Psi_1 \bl U_1^T ) + \bl I .  
\]
This covariance model is not identifiable unless restrictions 
are placed on the $\bs \Psi_k$'s. First, the eigenvectors of 
each $\bs \Psi_k$ are not identifiable: If 
$\bl \Psi_k =  \bl V_k \bs\Lambda_k \bl V_k^T$ is the 
eigendecomposition of $\bl \Psi_k$, then 
$ \bl U_k\bs\Psi_k \bl U_k^T = 
    \tilde {\bl U}_k  \bs \Lambda_k \tilde {\bl U_k}^T$, 
where $\tilde {\bl U}_k = \bl U_k \bl V_k^T \in \mathcal V_{r_k,n_k}$. 
Second, the scales of the $\bl \Psi_k$'s are not separately identifiable: 
For example, replacement of $(\bl \Psi_{k_1}, \bl \Psi_{k_2})$ with 
   $(c \bs\Psi_{k_1} , \bl \Psi_{k_2}/c) $  does not change the covariance matrix.  
With this in mind,   
$\bs \Psi$ is  parameterized as
$\bl \Psi  =  \tau^2 ( \bs\Lambda _K \otimes \cdots\otimes  \bs\Lambda_1)$
where $\tau^2>0$ and for each $k$,  $\bs \Lambda_k$  is an 
$r_k\times r_k$ diagonal matrix of positive entries that 
sum to 1.  

The parameters $\bs \Lambda_1,\ldots, \bs \Lambda_K$
can be interpreted in terms of the prior or penalty they induce over
the mode-specific eigenvalues
of the mean array $\bl M = \sigma \bl S \times \bl U$.
These eigenvalues
are often of interest in multiway data analysis
as they describe the extent to which the variation along a mode
can be attributed to a small set of orthogonal factors.   
To relate these eigenvalues to
the $\bs \Lambda_k$'s, recall that
$\bl M_{(1)} = \sigma \bl  U_1 \bl S_{(1)} \bl U_{(-1)}^T $, and so
$\bl M_{(1)} \bl M_{(1)}^T = \sigma^2 \bl U_1 \bl S_{(1)} \bl S_{(1)}^T 
    \bl U_{1}^T. $
Now $\bl S_{(1)}$ is equal  in distribution to
$\tau\bs \Lambda_1^{1/2} \bl Z \bs \Lambda_{-1}^{1/2}$, where
 $\bs\Lambda_{-1} = \bs \Lambda_K \otimes \cdots  \otimes 
   \bs \Lambda_2$ and $\bl Z$ is an $r_1\times r_{-1}$
  matrix of independent  standard normal entries. This gives
\begin{align*}
\Exp{ \bl M_{(1)} \bl M_{(1)}^T} &= 
\sigma^2  \tau^2 \bl U_1 \bs \Lambda_1^{1/2} 
 \Exp{ \bl Z \bs\Lambda_{-1}  \bl Z^T  } 
\bs\Lambda_1^{1/2}  \bl U_1^T  \nonumber \\ 
&= \sigma^2 \tau^2 \bl U_1  \bs \Lambda_1^{1/2} ( \text{tr}( \bs \Lambda_{-1} ) \bl I )  \bs \Lambda_1^{1/2}  \bl U_1^T \nonumber \\
&=   \sigma^2 \tau^2 \bl U_1  \bs \Lambda_1 \bl U_1^T ,  
\end{align*}
where the last calculation follows because the sum of the entries of each $\bs\Lambda_k$ is 1, making
 $\text{tr}(\bs\Lambda_{-1}) =1$.
Based on this calculation
for $\bl M_{(1)}$ (and analogous calculations for
  the other $\bl M_{(k)}$'s),
$\tau^2$  is seen to  be the expected squared magnitude of the
mean array $\bl M$ relative to the error variance $\sigma^2$, and  each
 $\bs \Lambda_k$ is the (scaled) diagonal eigenvalue  matrix 
of $\Exp{ \bl M_{(k)} \bl M_{(k)}^T}$. Additionally, if one or more 
of the diagonal elements of $\bs \Lambda_k$ are very close to zero, then 
$\bl M_{(k)}$ will be very close to a matrix of rank less than $r_k$.

An additional way to interpret the $\bs\Lambda_k$ parameters is  
in terms of a version of factor analysis for array-valued data. 
Under the heteroscedastic model for $\bl s$, the  marginal covariance of $\bl y$
takes the form of a convex combination of a reduced-rank positive semidefinite matrix
$ \bl U \bs \Lambda \bl U^T$
and the full-rank
matrix $\bl I$. 
This is similar to a factor analysis  model  
in which the covariance matrix
is equal to a reduced rank matrix, representing covariance due to
latent factors, plus
a full rank diagonal matrix representing measurement error. 
The fact that $\bl U$ and $\bs\Lambda$ are  separable
allows the factor analysis analogy to be applied to the modes of 
the array $\bl Y$ individually. For example,
considering  the expected 
sum of squares along the first mode
$\Exp{ \bl Y_{(1)} \bl Y_{(1)}^T }$, 
straightforward calculations show that 
\begin{align*}
\Exp{ \bl Y_{(1)} \bl Y_{(1)}^T }  &=  
 \sigma^2  (   
   \tau^2 \bl U_1 \bs\Lambda_1 \bl U_1^T  
   +  n_2\cdots n_K \bl I). 
\end{align*}
As with the covariance of $\bl y$,
this expectation of the  mode-1 
sum-of-squares matrix $ \bl Y_{(1)} \bl Y_{(1)}^T$
takes the form of a convex combination of a positive semidefinite matrix
$\bl U_1 \bs \Lambda_1 \bl U_1^T$ of reduced rank $r_1\leq n_1$
with eigenvalues $\bs \Lambda_1$
and a full-rank
diagonal matrix, as would be the case
in an ordinary factor analysis model that treated the rows of $\bl Y_{(1)}$
as variables and the columns as observations. 
One  difference between
ordinary factor analysis and this model is that the former presumes
independence along the columns of $\bl Y_{(1)}$, whereas this model
allows for dependence along each mode of $\bl Y$.
Another difference is that factor analysis permits  
a non-identity diagonal matrix in place of $\bl I$. 

\subsection{Simulation Study} 
A natural estimator of the reduced-rank mean array $\bl M$ based on 
the data array $\bl Y$ is the minimizer of the residual 
sum of squares $|| \bl Y - \bl M||^2$. 
If $K>2$
the least-squares estimator of $\bl M$ is not available in closed 
form, and so standard practice is to obtain a local minimizer  
$\hat {\bl M}_{\rm ALS}$ 
via an alternating least-squares (ALS) algorithm. 
The algorithm 
minimizes the sum of squares 
iteratively in the mode-specific eigenvectors of $\bl M$, 
a process that has been called
``higher order 
orthogonal iteration'' (HOOI) 
\citep{delathauwer_demoor_vandewalle_2000}.

One might anticipate that  estimates of the mean array $\bl M$ based on 
the homoscedastic model for $\bl S$, in which 
$\bl s\sim N_r(\bl 0, \tau^2 \bl I)$, 
will outperform $\hat {\bl M}_{\rm ALS}$
due to the ability of the former to 
shrink the values of $\bl S$ and the tendency of 
least-squares estimators
to  overfit, 
particularly for large values of $\bl r$.  
It might be further 
anticipated that the heteroscedastic covariance model for $\bl S$,
 in which
$\bl s\sim N_r(\bl 0, \tau^2(\bs \Lambda_K \otimes \cdots \otimes \bs \Lambda_1))$,
will  
outperform the homoscedastic model 
when $\bl r$ is chosen to be too large, 
as the heteroscedastic model allows for mode-specific shrinkage of 
the mean array towards estimates of lower rank. 
However,  such desirable performance in the case of a 
 misspecified rank
may come at the 
expense of poorer performance when the rank is correctly specified. 

These possibilities  were investigated with a simulation study comparing 
three different estimators of the mean array $\bl M$:
\begin{enumerate}
\item $\hat {\bl M}_{\rm ALS}$, 
obtained with 
the ALS algorithm; 
\item $\hat {\bl M}_{\rm HOM}$, the posterior mean  under the 
 homoscedastic model $\bl s \sim N_r(\bl 0, \tau^2 \bl I)$; 
\item $\hat {\bl M}_{\rm HET}$, the posterior mean  under
the heteroscedastic  model $\bl s \sim N_r(\bl 0, \tau^2 \bs \Lambda_K\otimes \cdots \otimes 
  \bs \Lambda_1 ).$ 
\end{enumerate}
The Bayes estimator $\hat {\bl M}_{\rm HOM}$ was obtained 
using a conjugate inverse-gamma$(\nu_0/2,\tau_0^2/2)$ prior for $\tau^2$, 
where $\nu_0=1$ and $\tau_0^2 =\prod_{k=1}^K n_k/r_k$. This value of $\tau_0^2$ 
makes the expected prior magnitude of the mean array 
equal to that of the error, so that 
$\Exp{ ||\bl M||^2 }  = \Exp{ ||\bl E||^2 } $ \emph{a priori}. 
The Bayes estimator 
 $\hat {\bl M}_{\rm HET}$ was obtained under 
a prior on $(\tau^2, \bs \Lambda_1,\ldots, \bs \Lambda_K)$ 
in which  $\tau^2$ has an 
inverse-gamma$(1/2,\tau_0^2/2)$ distribution  
and the diagonal elements of each 
$\bs\Lambda_k$ are uniform on the $r_k$-dimensional simplex.  
The value of $\tau^2_0=\prod_{k=1}^K n_k$ was chosen so that 
 $\Exp{ ||\bl M||^2 }  = \Exp{ ||\bl E||^2 } $ {\it a priori}, as with 
the prior  used to obtain
$\hat {\bl M}_{\rm HOM}$. 
The uniform priors on the $\bs\Lambda_k$'s are not conjugate, 
and so the Markov chain for posterior estimation in this model
relies on a Metropolis-Hastings update for these parameters. 

Three-dimensional data arrays $\bl Y \in \mathbb R^{60\times 50\times 40}$
were simulated 
according to the following procedure:
For a given rank vector
$\bl r_0=(r_{01},r_{02},r_{03})$,  
\begin{enumerate}
\item simulate $\bl U_k \sim \text{uniform}(\mathcal V_{r_{0k},n_k} ) $
for each $k\in \{1,2,3\}$;
\item simulate  $\bl s \sim N_r(\bl 0, \psi\times \left ( \prod_{k=1}^{K} r_{0k}^2\right )^{-1/3}\times \bl I)$;
\item let $\bl M =
  \bl S \times \{ \bl U_1,\ldots, \bl U_K \} $, where   
  $\text{vec}(\bl S) = \bl s$; 
\item let $\bl Y =  \bl M + \bl E $,
where $\bl E$ has i.i.d.\ standard normal entries.
\end{enumerate}
Data were generated under two values of $\bl r_0$ and
two values of $\psi$ for a total of four different conditions.
The
values of $\bl r_0$ included 
a ``low-rank'' condition $\bl r_0=(6,5,4)$ 
and 
a ``high-rank'' condition $\bl r_0=(30,25,20)$, 
and the 
values of $\psi$ included  a  
``low-signal'' condition $\psi=1000$  and a ``high-signal'' 
condition $\psi=2000$. 
Ten datasets were generated under each of these four conditions, 
for a total of forty simulated datasets. 
For each dataset, 
$\hat {\bl M}_{\rm ALS}$, $\hat {\bl M}_{\rm HOM}$ 
and $\hat {\bl M}_{\rm HET}$ were obtained
with the assumed rank $\bl r$ equal to the
true rank $\bl r_0$. 
Each Bayesian estimate was obtained via 11,000 iterations of the 
MCMC
algorithm described in the previous section.
The first 1000 iterations of each Markov chain were dropped to 
allow for convergence to the stationary distribution, and parameter 
values were saved every 10th iteration thereafter, resulting 
in 1000 simulated values of $\bl M$ with which to approximate its 
posterior mean. Convergence and mixing of the 
Markov chains were monitored via traceplots of 
the simulated values of 
$\sigma^2$ and 
$\tau^2$,  as well as their effective sample sizes, which
roughly measure the approximation variability 
of the posterior mean estimates relative to those that would be obtained 
from independent Monte Carlo simulations.
Effective sample sizes for $\sigma^2$ and $\tau^2$ were above 300
for all scenarios and datasets, 
and close to  half  the Markov chains attained the maximum possible value of 
1000.

\begin{table}
\begin{center}
\begin{tabular}{||r||c|c||c|c||} \hline
 rank &\multicolumn{2}{c||}{$\bl r_0=(6,5,4)$}&\multicolumn{2}{c||}{$\bl r_0=(30,25,20)$}\\
  \hline
signal  &  low  & high  & low  & high  \\ \hline  \hline
RSE$( \hat {\bl M}_{\rm ALS} )$ &  0.195 & 0.088 & 0.848 & 0.379 \\
RSE$( \hat {\bl M}_{\rm HOM} )$ &  0.165 & 0.082 & 0.485 & 0.280 \\
RSE$( \hat {\bl M}_{\rm HET} )$ & 0.165& 0.082 & 0.489 & 0.281 \\ \hline
\end{tabular} 
\end{center}
\caption{Relative squared estimation errors.}
\label{tab:rse}
\end{table}

For each estimator 
and each simulation condition,  a
relative squared estimation error (RSE) 
was computed
by averaging 
the value of 
$|| \bl M -  \hat{\bl M}||^2/||\bl M||^2$ 
across the 10 datasets. 
These values are given in Table \ref{tab:rse}. 
Note that $\hat {\bl M}_{\rm HOM}$ is to some extent 
an ``oracle'' estimator, in that it is based on a prior 
distribution that was used to simulate the data 
(although $\hat {\bl M}_{\rm HOM}$ requires estimation of $\tau^2$). 
Nevertheless, 
in the low-rank case $(\bl r_0=(6,5,4))$, the two Bayes 
estimators performed nearly identically in terms of RSE, and 
the ALS estimator performed slightly worse. In terms of variability across 
datasets, $\hat {\bl M}_{\rm HOM}$ outperformed 
  $\hat {\bl M}_{\rm ALS}$  for all datasets,  
and outperformed $\hat {\bl M}_{\rm HET}$ in  10 of the 20 datasets.  
The story is similar for the 20 high-rank datasets $(\bl r_0=(30,25,20))$, 
except that  ALS performs 
more poorly in this case than in the low-rank case, presumably 
because of the much larger number of parameters and the 
general tendency 
of least-squares estimators to overfit the data. Regarding this, 
the residual squared error  $||\bl Y-\hat {\bl M}||^2$ 
was lower for the ALS estimator
than the Bayes estimators
across all datasets and scenarios. 

For the same 40 simulated datasets, 
estimates $\hat {\bl M}_{\rm ALS}$,  $\hat {\bl M}_{\rm HOM}$
and  $\hat {\bl M}_{\rm HET}$ 
were also obtained
using a fitted rank of 
$\bl r =2\times \bl r_0$, that is, twice 
the actual rank of $\bl M$. Note that in the 
high-rank scenario the fitted rank is $\bl r=(60,50,40)$, which is the dimension of the data array. In this case, the  estimates 
are of full rank and so in particular the ALS estimate  is 
simply $\bl Y$. 
Also, the Bayes estimates in this full rank case were obtained using a
proper gamma$(1/2,1/2)$ prior distribution for 
$\sigma^2$ to guarantee the propriety of the posterior 
(recall the discussion in Section 2). 
Relative squared errors (RSEs) for these misspecified-rank
estimators 
 are given in Table \ref{tab:rse_of}. 
Not surprisingly, $\hat {\bl M}_{\rm ALS}$ performs poorly across  
all scenarios, and roughly 4 to 6 times worse than it does
when the rank is correctly specified. 
The Bayes estimator  $\hat {\bl M}_{\rm HOM}$ 
performs reasonably well in the low-rank scenario, 
but roughly 3 times worse than 
it does in the high-rank scenario with correctly specified rank. 
In contrast, the performance of $\hat {\bl M}_{\rm HET}$  
with  a misspecified rank is nearly identical to its performance 
with  a correctly specified rank. 
This suggests that the 
heteroscedastic model for $\bl S$ 
is able to shrink the estimate of $\bl M$ towards arrays of the correct rank.

\begin{table}
\begin{center}
\begin{tabular}{||r||c|c||c|c||} \hline
 rank &\multicolumn{2}{c||}{$\bl r_0=(6,5,4)$}&\multicolumn{2}{c||}{$\bl r_0=(30,25,20)$}\\
  \hline
signal  &  low  & high  & low  & high  \\ \hline  \hline
RSE$( \hat {\bl M}_{\rm ALS} )$ &  0.855 & 0.404 & 4.840 & 2.420 \\
RSE$( \hat {\bl M}_{\rm HOM} )$ &  0.260 & 0.141 & 1.364  & 0.840 \\
RSE$( \hat {\bl M}_{\rm HET} )$ & 0.166 & 0.082 & 0.495 & 0.284  \\ \hline
\end{tabular}
\end{center}
\caption{Relative squared estimation errors 
 when the fitted rank is twice that 
of $\bl r_0$.}
\label{tab:rse_of}
\end{table}

\begin{figure}[ht]
\centerline{\includegraphics[width=6.5in]{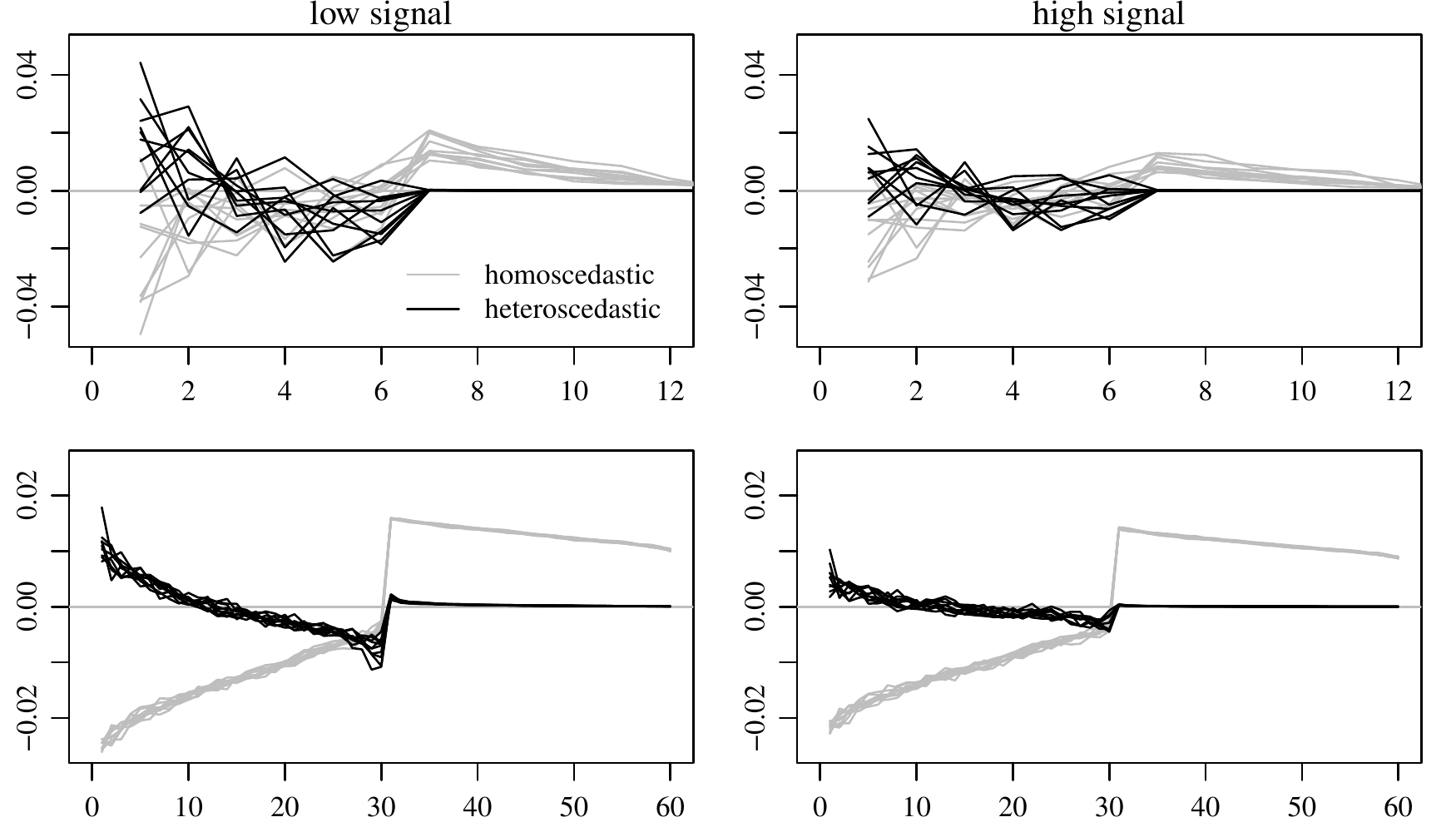}} 
\caption{Difference of eigenvalues between 
$\hat {\bl  M}_{(1)}\hat{\bl  M}_{(1)}^T$   and 
   ${\bl  M}_{(1)}{\bl  M}_{(1)}^T$. 
Estimates in the first row are based on 
a true rank of $\bl r_0=(6,5,4)$ and  a fitted rank 
of $\bl r=(12,10,8)$.  
Estimates in the second row are based on 
a true rank of $\bl r_0=(30,25,20)$ and  a fitted rank 
of $\bl r=(60,50,40)$. }
\label{fig:espectra}
\end{figure}

This is explored further in Figure \ref{fig:espectra}. 
For each Bayesian estimate $\hat {\bl M}$ obtained with a misspecified 
rank, its mode-1 matricization $\hat {\bl M}_{(1)}$ was constructed
and the normalized eigenvalues of   
$\hat {\bl  M}_{(1)}\hat{\bl  M}_{(1)}^T$ were computed, 
from which the 
normalized
eigenvalues of ${\bl  M}_{(1)}{\bl  M}_{(1)}^T$ were subtracted off, 
where $\bl M_{(1)}$  is 
the mode-1 matricization of the true mean array $\bl M$. 
These eigenvalue differences are plotted across datasets and 
conditions in Figure \ref{fig:espectra}. 
For example, the plot in the upper-left corner of the figure 
shows results under the low-signal low-rank condition, 
for which the true rank is $\bl r=(6,5,4)$ but the fitted rank is
$\bl r=(12,10,8)$. Each black line corresponds to the 
eigenvalues of $\hat {\bl  M}_{(1)}\hat{\bl  M}_{(1)}^T$ 
obtained under the heteroscedastic model
minus the eigenvalues of ${\bl  M}_{(1)}{\bl  M}_{(1)}^T$,
for one of the 10 simulated datasets. 
The gray lines correspond to the analogous differences under the 
homoscedastic model. 
The results indicate that the homoscedastic
model
generally underestimates  non-zero eigenvalues 
and substantially overestimates zero eigenvalues. 
In contrast, the heteroscedastic model generally 
does a very good job of estimating the zero eigenvalues 
as being  very nearly zero.  However, 
for the non-zero eigenvalues, the estimated eigenvalues for the 
the heteroscedastic model  are somewhat too ``steep'', 
overestimating the true 
large non-zero eigenvalues and 
underestimating the small non-zero eigenvalues.  
A larger signal appears to ameliorate these biases, as 
the differences between estimated and true eigenvalues is diminished in 
going from the low-signal to the high-signal scenario. 
However, the presence of such biases suggests 
exploration of more complex adaptive penalties or hierarchical priors,
 i.e.\  ones that could more flexibly adapt to the shape of the eigenspectra 
in the observed data. For example, 
a beta$(a,b)$ prior over the diagonal elements of $\bs\Lambda_k$ 
could be used instead of the uniform prior. 
However, in the absence of prior 
information about the eigenspectra, the values of $a$ and $b$ would 
need to be obtained from the data. 
Such an empirical Bayes approach would be similar in spirit to the 
two-parameter matrix regularizer of 
\citet{josse_sardy_2013}.

\section{A scale-free Tucker decomposition model}
In this section the TDM is extended
in order to analyze 
data arrays 
for which the assumption of normally distributed errors 
is inappropriate. 
The approach presented is based upon a transformation model
in which the observed data array is modeled as  
an unknown increasing function of a latent array that follows 
a normal TDM. The model fitting procedure 
provides parameter estimates that are invariant to 
monotonic transformations of the data array, thereby 
giving a ``scale-free'' TDM.
This approach is  motivated and illustrated with an analysis 
of discrete multivariate data on relations between countries in the 
year 2012. 


\subsection{Data description}
The motivating application of this section is
to obtain a low-rank representation of 
relational data on actions between countries,
obtained from the GDELT project (\href{http://gdelt.utdallas.edu/}{gdelt.utdallas.edu}). The data analyzed 
consist of a weekly summary of
20 different types of actions
between the 30 most active countries in the GDELT database in
2012. These data can be represented as a $30\times 30 \times 20 \times 52$
four-way array $\bl Y$, with entries
$\{ y_{i,j,k,t} : 1 \leq i , j \leq 30, i\neq j, 1\leq k\leq 20, \ 
                    1\leq t \leq 52 \}$
where $y_{i,j,k,t}$ is the number of days in
week $t$ in which
country $i$ took action $k$ with country $j$ as the target. 
The types of actions 
include ``positive'' 
actions such as diplomatic cooperation and the provision of 
aid, as well as ``negative'' actions such as the expression 
of disapproval, military threats and military conflict.  
More details on the action types, as well as a list of the 30 countries in 
the data array, are provided in the Appendix.  
Figure \ref{fig:exnets} provides a graphical summary of the array $\bl Y$
for four of the twenty action types.  To construct this figure, 
counts for each of the four action types 
between each ordered pair of countries
were summed across the 
52 weeks of the year and then dichotomized, so that a link between 
two countries indicates the presence of the action type for 
at least one day of the year. 

\begin{figure}[ht]
\centerline{\includegraphics[width=6.5in]{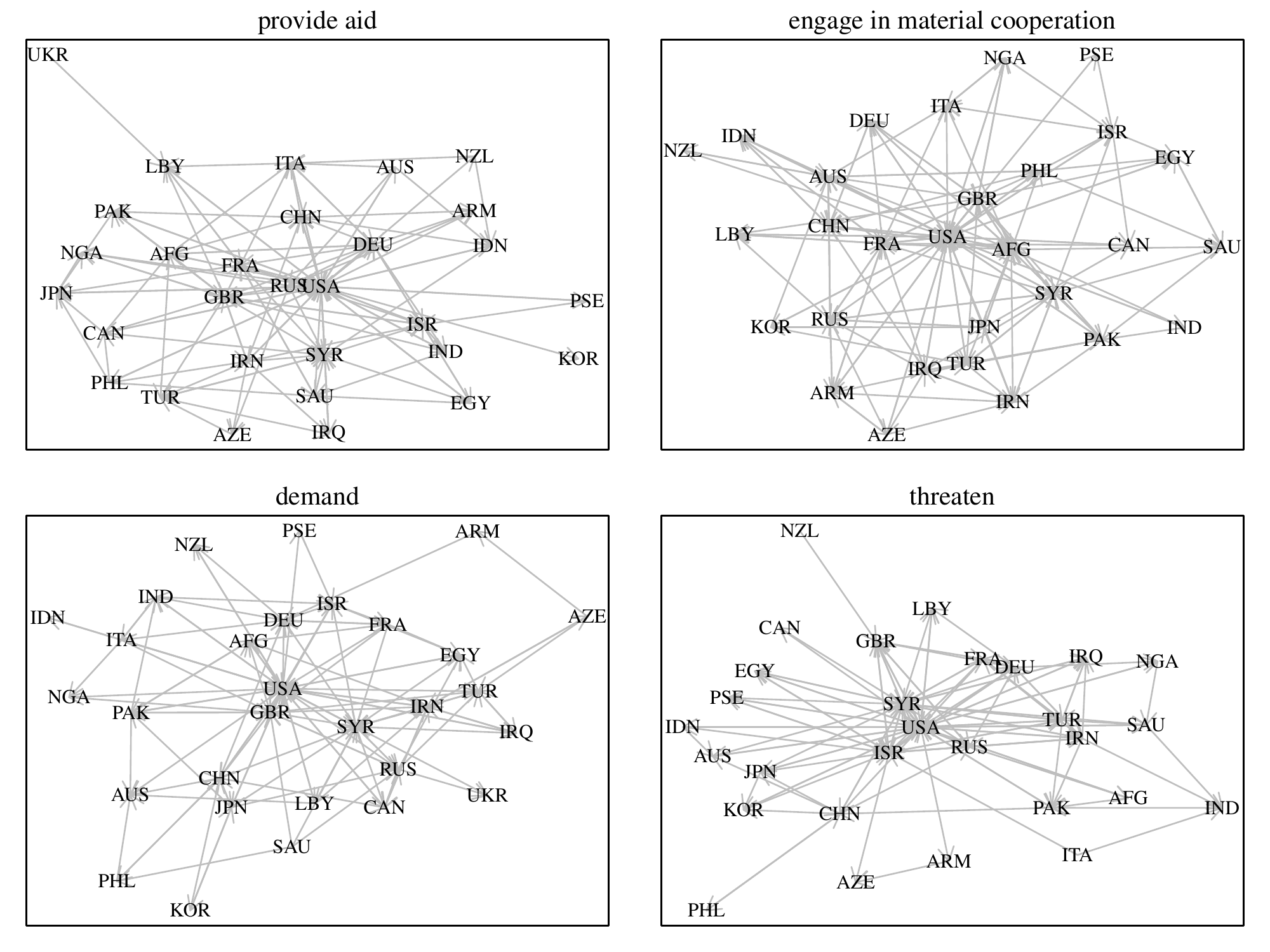}}
\caption{Networks corresponding to four of the twenty action types.}
\label{fig:exnets}
\end{figure}

The  data array $\bl Y$  has nearly one million entries but is very sparse,
with just over 2\% of the entries being non-zero.
This sparsity varies  by action type
from a high of  about 12\% for the action ``consult'' to
a low of less than .01\% for the action ``use unconventional mass violence.''
Sparsity also varies considerably by country: The first panel of 
Figure \ref{fig:dplots} plots outdegrees and indegrees of each country, 
computed (for country $i$) as 
$\sum_{jkt} y_{i,j,k,t}$ and $\sum_{jkt} y_{j,i,k,t}$, respectively. 
These two measures of activity are highly correlated across countries, 
with Syria being somewhat of an outlier, being the target of 
more actions than it initiates. 
Additionally,  
the counts for each 
action are highly skewed: There are more counts of 
zero than counts of one, more counts of one than counts of two, 
and so on.  This is illustrated in the second panel of 
Figure \ref{fig:dplots}, which gives the empirical distribution of 
the nonzero  entries of $\bl Y$. 

\begin{figure}[ht]
\centerline{\includegraphics[width=6in]{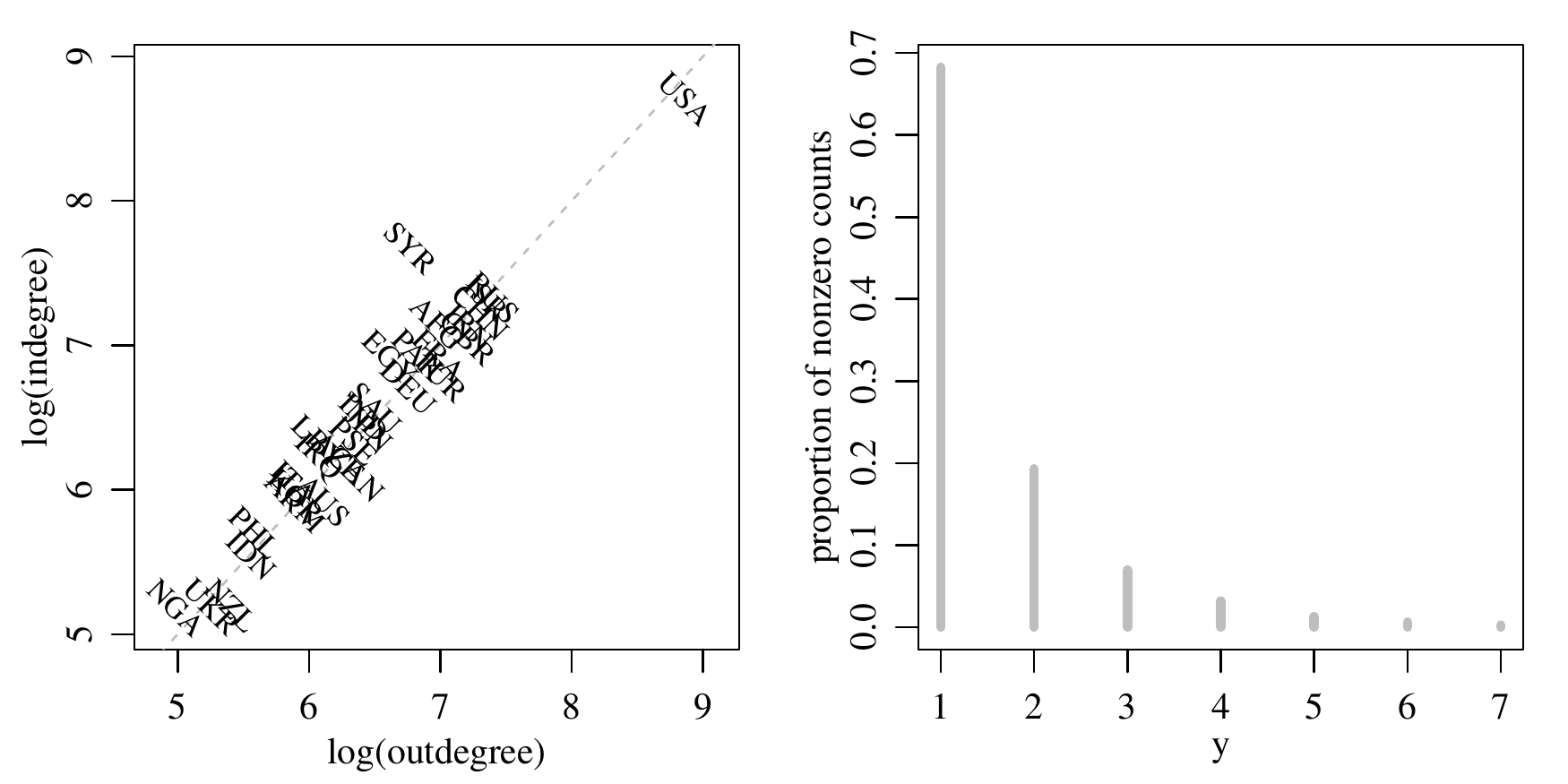}}
\caption{Descriptive data plots. The left panel shows 
country-specific outdegrees and indegrees on the log scale. 
The right panel gives a histogram of the non-zero action counts.}
\label{fig:dplots}
\end{figure}

\subsection{Scale free TDM}
Existing array decomposition 
methods applied directly to these data would
be problematic for several reasons.
One particular issue in applying matrix or array
decomposition methods to relational datasets is that
self-relations are typically undefined, that is,
$y_{i,i,k,t}$ is not defined for any $i$, $k$ or $t$. 
This issue can be addressed via an alternating least-squares 
algorithm that iterates between 
fitting a  reduced-rank model and 
replacing any missing values with 
fitted values
(see, for example, \citet{ward_hoff_2007} for details on 
such an algorithm applied to matrix-valued relational data). 
A more serious problem is that the discrete or ordinal nature of
many relational datasets  makes
least-squares methods of limited use. 
For example, 
as will be illustrated at the end of this section, 
a reduced rank representation of 
the GDELT data array $\bl Y$  
obtained via alternating least squares
generally 
represents the largest data values at the expense of 
other interesting features of  the data. 

While the normal TDM model presented in the previous section may 
not be appropriate
for ordinal or discrete data,
the normal model can be extended to accommodate
such data via a latent variable formulation, in 
which the entries of $\bl Y$
are modeled as a non-decreasing function of the elements
of a latent array $\bl Z$ that follows the Tucker decomposition model. 
If the elements of $\bl Y$ take on a 
known finite number of 
possible values, then 
such an approach can be viewed as similar to an 
ordered probit model. 

In many datasets
one of the indices of the array $\bl Y$ represents
variables that may be best evaluated on different scales. 
For example,  the large heterogeneity in
sparsity between the 20 different action types in the GDELT dataset
suggests modeling the different types  on different scales. 
As another
example,
consider an  $ n\times n \times 2$ relational array  where
$y_{i,j,1} $ is the number of emails sent from person $i$ to
person $j$, 
and $y_{i,j,2}$ encodes an evaluation 
of $i$'s  friendship to $j$ on an 
ordinal scale.
In such a case, it may not make sense to model
$y_{i,j,1}$ and $y_{i,j,2}$ as the same transformation of
the latent variables
$z_{i,j,1}$ and $z_{i,j,2}$. In particular, the number of levels
of the two variables may be different.
For cases such as these, a more appropriate transformation model 
may be one with 
with variable-specific transformations, so that 
\begin{align}
\bl Z & = \bl S \times \{ \bl U_1 , \cdots ,  \bl U_K  \}  + \bl E \\ 
 \text{vec}(\bl E) & \sim  N_n(\bs 0 , \bl I)  \nonumber \\
y_{\bs i,j} & = g_{j}( z_{\bs i,j} ),   \nonumber
\label{eq:ctsvd}
\end{align}
where $\bs i\in \{1,\ldots, n_1\}\times \cdots \times \{ 1,\ldots, n_{K-1}\}$, 
$j\in \{1,\ldots, n_K\}$, and 
for notational convenience  
the variables to be modeled on different scales are indexed by the 
$K$th mode of the array. 
Note that the scale parameter $\sigma$ from the TDM
in the previous sections
would be confounded with the transformations $g_1,\ldots, g_{n_K}$, 
and so can
be set 
to 1.

In the case that the transformations $g_1,\ldots, g_{n_K}$
are nuisance parameters, 
scale-free estimation of $(\bl S , \bl U)$ can be obtained 
using a rank likelihood $L_R$, defined as 
\[   L_R( \bl S , \bl U : \bl  Y ) = 
     \Pr( \bl Z \in R(\bl Y)  | \bl S, \bl U), \]
where  $R(\bl Y)$ is the set of  $\bl Z$-values  
consistent with the observed data $\bl Y$ and the fact 
that the functions $g_1,\ldots, g_{n_K}$ are non-decreasing. 
This set can be expressed as 
\[ 
 R(\bl Y) =  \{ \bl Z : 
 \max \{  z_{\bs i',j} : y_{\bs i',j}<y_{\bs i,j}  \} < 
        z_{\bs i,j} <  \min \{  z_{\bs i',j} : y_{\bs i,j}<y_{\bs i',j} \}\}.
\]
A feature of estimates obtained from the rank likelihood        
is that they are scale-free: 
The set $R(\bl Y)$ is invariant to strictly increasing 
transformations of the data, and therefore so is the rank likelihood.

While maximum likelihood estimation using the rank likelihood is 
generally computationally intractable,
Bayesian inference using this likelihood 
is feasible via the Gibbs sampler
(see 
\citet{hoff_2007a} and 
\citet{hoff_2008b} for applications of the 
rank likelihood to semiparametric copula and regression 
models, respectively).  
Under a prior distribution for $(\bl S, \bl U)$ from the previous 
section, 
posterior estimates for this scale-free  TDM
can be obtained via a simple extension of the previous algorithm.
The extended algorithm can be roughly understood as follows:
If $\bl Z$ were observed, parameter estimates
could be obtained from  the MCMC algorithm for the normal TDM. 
As $\bl Z$ is not observed,  the algorithm requires 
 additional steps in order 
to integrate over the  possible values of $\bl Z$. 
This can be done by simulating
values of the elements of $\bl Z$
from their full conditional distributions
at each step of the Markov chain.
Specifically, 
posterior approximation for this scale-free TDM can proceed
by iterating the following steps:
Given current values $(\bl Z, \bl S, \bl U  )$,
\begin{enumerate}
\item update $(\bl S, \bl U )$
as in the case of the normal TDM,  with $\bl Z$ taking on the role of $\bl Y$;
\item update the elements of $\bl Z$ given $\bl Y$, $\bl S$ and $\bl U$ as follows:
\begin{enumerate}
\item compute $\bl M = \bl S \times \{ \bl U_1,\ldots, \bl U_K\}$;
\item simulate each  $z_{\bs i, j }$ from the
constrained normal$(m_{\bs i, j}, 1)$  distribution, constrained so that 
\[    
\max \{  z_{\bs i',j} : y_{\bs i',j}<y_{\bs i,j}  \} < 
        z_{\bs i,j} <  \min \{  z_{\bs i',j} : y_{\bs i,j}<y_{\bs i',j} \}. 
\]
\end{enumerate}
\end{enumerate}
Iteration of steps 1 and 2 generates a Markov chain, 
samples from which approximate 
the  posterior distribution proportional to 
$L_R(\bl S, \bl U : \bl Y) \times \pi(\bl S, \bl U)$. 
As mentioned above, parameter estimates obtained 
from this posterior distribution are invariant 
to monotonic transformations of each variable along the 
$K$th mode of the array.   
For this reason, this estimation procedure and the resulting estimates 
can be referred to as a 
scale-free Tucker decomposition (SFTD). 


\subsection{Analysis of GDELT data}
A rank $\bl r=(4,4,4,4)$ 
representation of the GDELT data was obtained from the 
SFTD procedure described above, 
using the heteroscedastic prior described in Section 4 and 
modeling the 20 action 
types on different scales. 
A rank of 4 for each mode was chosen because of the substantial 
amount of heterogeneity in the 
outdegrees and indegrees 
as displayed in the first panel of
Figure \ref{fig:dplots}. 
A standard approach to representing such heterogeneity 
would be with an additive model in which the entries of $\bl M$ are 
expressed as  
the sum of mode-specific effects, for example 
$m_{i,j,k,t} =a_i+b_j+c_k+d_t$. 
Such an 
additive effects model has a rank of $(2,2,2,2)$. 
A rank $(4,4,4,4)$ approximation was fit  to $\bl Y$ in order to capture 
the rank $(2,2,2,2)$ additive effects along with
two additional dimensions of 
non-additive data patterns, 
which are shown below. 

The MCMC algorithm described above was run for 55,000 iterations. 
The first 5,000 iterations were dropped to allow for convergence
to the stationary distribution, 
and parameter values were saved every 10th iteration thereafter. 
This resulted in 5,000 simulated values of the parameters with which 
to approximate posterior quantities of interest. 
Mixing of the Markov chain was evaluated with traceplots 
and effective sample sizes 
of $\tau^2$ and the 
eigenvalue parameters $\bl \Lambda_1,\ldots , \bl \Lambda_4$. 
The effective sample size
for $\tau^2$ was 1197.
Effective sample sizes for the eigenvalues ranged between 
371 and 1266, with a mean of 678. 
Traceplots of some of the eigenvalues are 
plotted in Figure \ref{fig:mcmcdiag}. The 
first eigenvalues of the third and fourth modes 
(corresponding to action type and week) are close to one 
with high posterior probability, meaning that $\bl M_{(3)}$ and 
$\bl  M_{(4)}$ are both close to being rank-1 matrices. 
Eigenspectra of the first and second modes 
(corresponding to initiators and targets of the actions)
were more evenly distributed. For both of these two modes, 
the first two eigenvectors predominantly represented 
the heterogeneity in outdegrees and indegrees. To examine 
non-additive patterns in the data, the posterior 
mean array $\hat {\bl M}$ was centered along each index of each mode, creating 
an array  $\tilde {\bl M}$ representing the non-additive patterns in the data. 
\begin{figure}[ht]
\centerline{\includegraphics[width=5.5in]{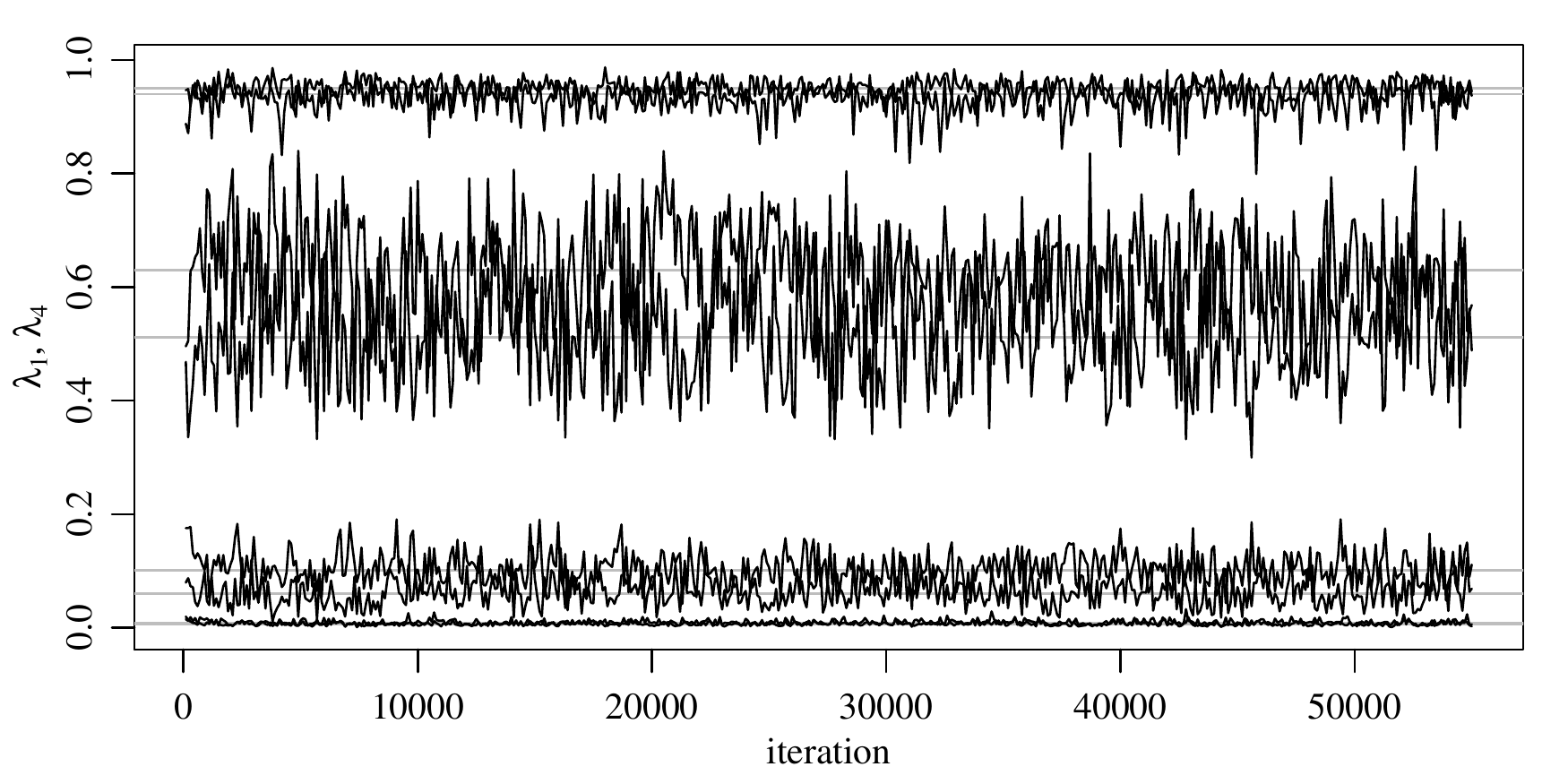}}
\caption{Traceplots of $\lambda_{1,k}$ and $\lambda_{4,k}$ for $k\in
\{1,2,3,4\}$.  }
\label{fig:mcmcdiag}
\end{figure}

The first two left singular vectors of $\tilde {\bl M}_{(1)}$, 
 $\tilde {\bl M}_{(2)}$ and $\tilde {\bl M}_{(3)}$ 
are displayed in Figure \ref{fig:fplot_pm}. 
The first two plots indicate strong geographic patterns
in the first two modes of $\tilde {\bl M}$. 
These patterns indicate that, after accounting 
for additive effects, 
countries that have similar patterns of activity in the dataset
are typically close to one another geographically.  
The converse is not generally true:  
PSE and ISR are far apart from SYR, IRQ and IRN on the plot, 
indicating heterogeneity in the dataset that is non-geographic. 
The third plot in Figure \ref{fig:fplot_pm}
displays the singular vectors of $\tilde {\bl M}_{(3)}$
 corresponding to 
the different action types. Plotting symbols ``+'' 
and ``-'' are used to indicate 
actions that are categorized as ``positive'' or ``negative'' respectively
(a list and categorization  of the action types are given in the appendix). 
The singular vectors of $\tilde {\bl M}_{(3)}$ distinguish somewhat the two types of actions, but there 
is considerable overlap. This is not too surprising, since 
countries that interact frequently with each other generally 
relate both positively and negatively during the course of the year. 

\begin{figure}[ht]
\centerline{\includegraphics[width=6.75in]{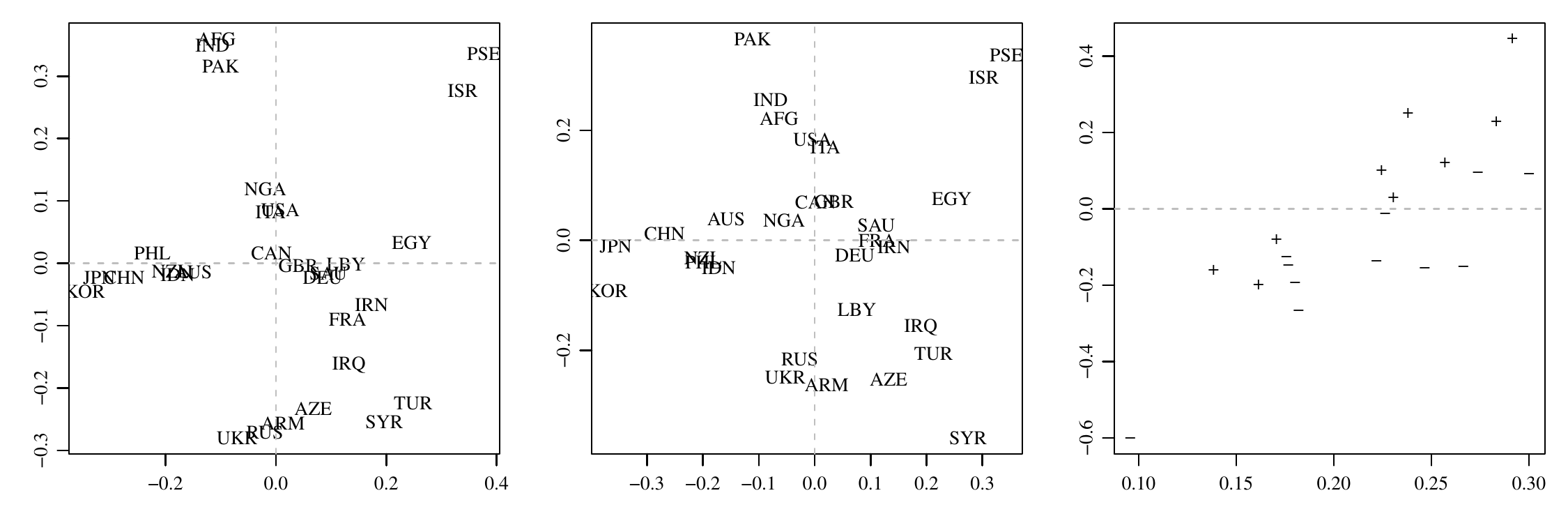}}
\caption{Plots of the first two left singular vectors of 
  $\tilde {\bl M}_{(1)}$,
  $\tilde {\bl M}_{(2)}$ and $\tilde {\bl M}_{(3)}$, 
 from the SFTD of $\bl Y$.}
\label{fig:fplot_pm}
\end{figure}

The utility of the SFTD in comparison to a least-squares 
approach can be seen by contrasting this scale-free representation of $\bl Y$ 
given in 
Figure \ref{fig:fplot_pm} to an analogous least-squares representation 
shown in 
in Figure \ref{fig:fplot_ls}. 
This plot 
gives the first two singular vectors of the first three modes
 of $\tilde {\bl M}_{\rm ALS}$, where $\tilde {\bl M}_{\rm ALS}$ was constructed as 
with the SFTD
except using 
a rank (4,4,4,4) alternating least-squares 
approximation $\hat {\bl M}_{\rm ALS}$ to $\bl Y$ instead of the 
posterior mean array $\hat {\bl M}$. 
The least squares approach is primarily 
identifying the countries that have the most
number of 
data values of 7 (the highest value possible), 
at the expense of representing the other patterns in the data.  
For example, the first singular vectors of the both the first- and 
second-mode matricizations of $\tilde {\bl M}_{\rm ALS}$
are essentially 
devoted to distinguishing the USA from the other countries. 

\begin{figure}[ht]
\centerline{\includegraphics[width=6.75in]{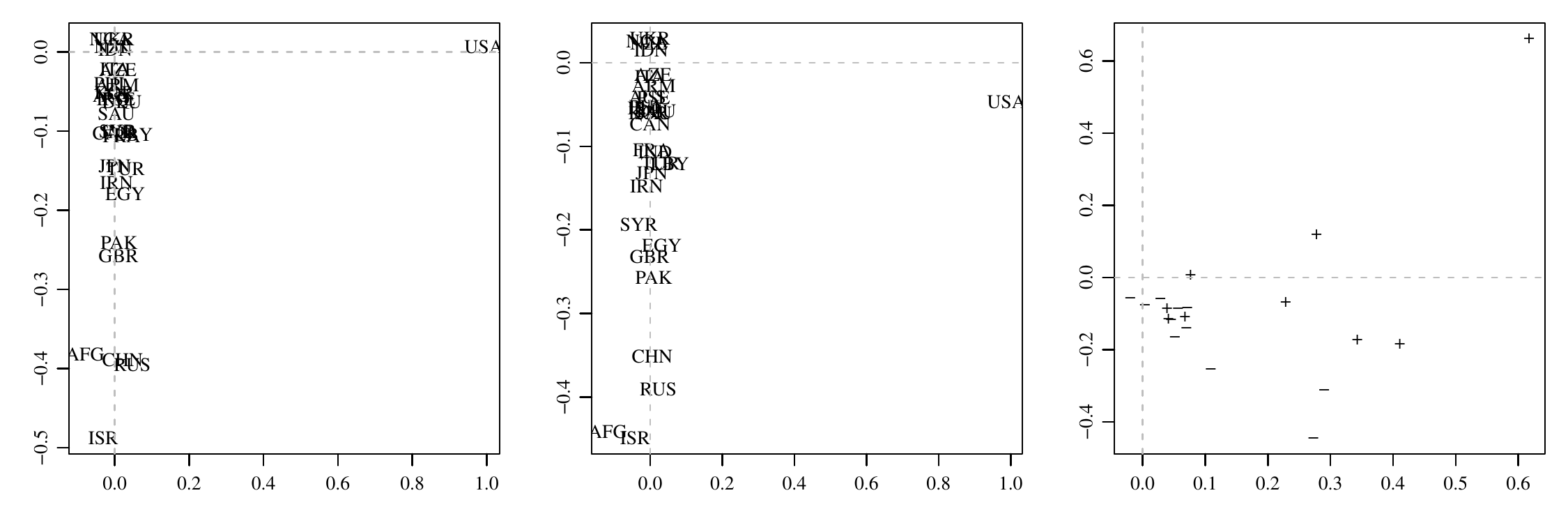}}
\caption{Plots of the first two left singular vectors of     
  $\tilde {\bl M}_{{\rm ALS}(1)}$,
  $\tilde {\bl M}_{{\rm ALS}(2)}$ and $\tilde {\bl M}_{{\rm ALS}(3)}$.  }
\label{fig:fplot_ls}
\end{figure}

The posterior mean array 
$\hat {\bl M}$ and the least squares representation $\hat {\bl M}_{\rm ALS}$ 
can also be evaluated
in terms of how well they represent the rank ordering of the values of $\bl Y$. 
This is done by computing 
Kendall's $\tau$, a scale-free measure of association, 
between the entries of $\bl Y$ and each of the two 
low-rank representations $\hat {\bl M}$ and $\hat {\bl M}_{\rm ALS}$. 
This is done  separately for each of the 20 action types 
in order to evaluate any heterogeneity in performance. 
As shown in Figure  \ref{fig:kdtcomp}, the  SFTD representation
has a higher degree of association with the ranks of $\bl Y$ 
than the least-squares representation for all action types. 
This is perhaps not too surprising - the SFTD is inherently 
scale-free, and so $\hat {\bl M}$ is only representing 
information about the rank ordering of the entries of $\bl Y$. 
In contrast, $\hat {\bl M}_{\rm ALS}$ must also represent differences in 
magnitude. For these highly skewed data, 
a good representation 
of large  differences in magnitude comes at the cost of a poorer representation 
of small differences, which constitute most of the differences in the entries of $\bl Y$.

\begin{figure}[ht]
\centerline{\includegraphics[width=2.75in]{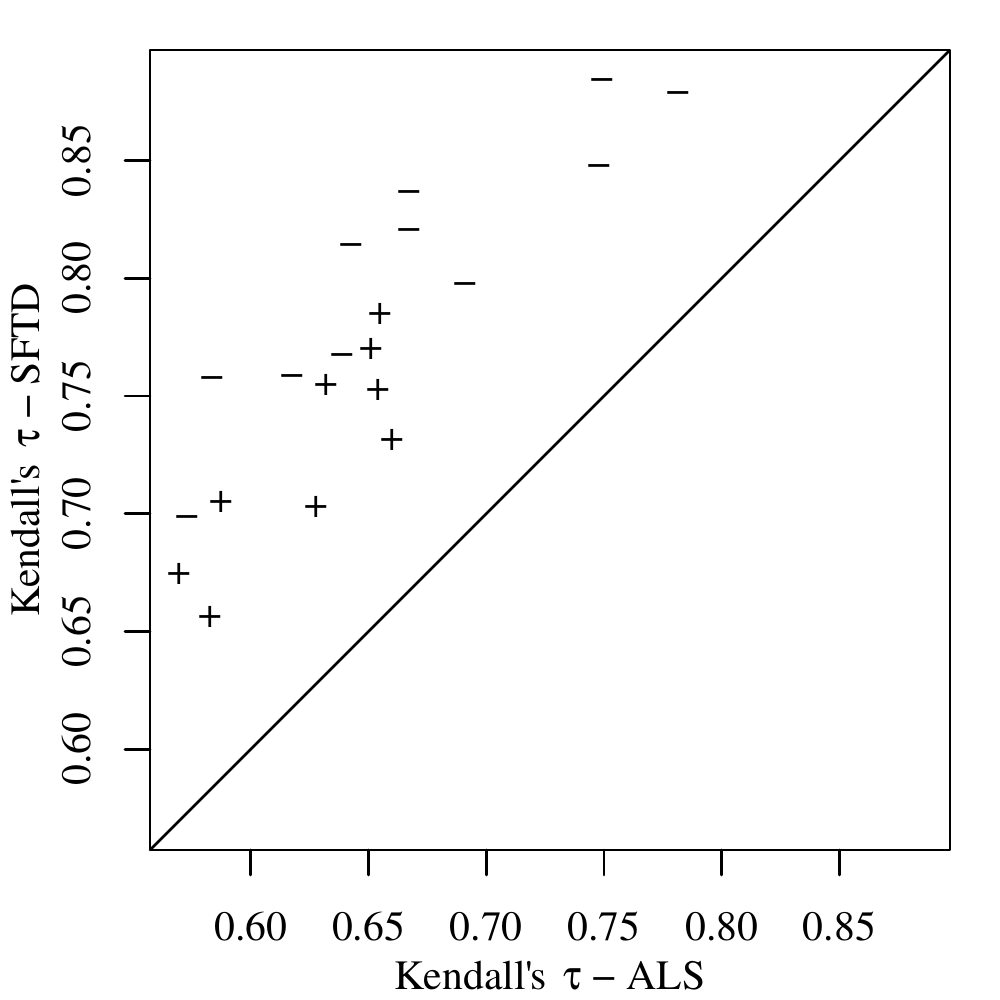}}
\caption{Kendall's $\tau$ measure of association between 
$\bl Y$ and  $\hat {\bl M}$ (vertical axis)  
and $\hat {\bl M}_{\rm ALS}$ (horizontal axis), by action type. 
Positive and negative actions are plotted as ``+'' and ``-''
respectively. }
\label{fig:kdtcomp}
\end{figure}


\section{Discussion}
While the objectives of an 
array-valued data analysis 
may be primarily descriptive, 
model-based approaches may be appealing for a variety of 
reasons. For example, regularized data descriptions 
may be obtained using 
model-based Bayesian procedures, with the prior acting as a penalty term.
This article has  developed 
a parameterization of the normal Tucker decomposition
model that allows for scale-equivariant and  orthogonally-equivariant
estimates and data descriptions, while still allowing for 
penalization of mode-specific singular values. 
Such regularized estimates can greatly improve upon least-squares estimates in 
situations where the data array is equal to a reduced-rank mean 
array plus noise.
Another benefit of the model-based approach is its 
extensibility to a variety of different data types and data analysis 
scenarios.
For example, the semiparametric transformation model 
developed in Section 5 provides a scale-free reduced-rank  
representation for data arrays that consist of discrete, ordinal or 
other types of measurements for which a least squares criterion 
is not appropriate.

A useful extension of the model would be to data analysis 
situations in which it is desired to account for
known explanatory factors or patterns in the data. 
For example, 
one extension of the model used to analyze the 
GDELT data in Section 5 takes the form
\begin{align*}
\bl Z & = \langle \bl X , \bl B \rangle  +  \bl S \times \{ \bl U_1 , \cdots ,  \bl U_K  \}  + \bl E \\ 
 \text{vec}(\bl E) & \sim  N(\bs 0 ,  \bs\Sigma(\rho)  \otimes\bl I \otimes\bl I \otimes\bl I ) 
\end{align*}
where $\bl X$ and $\bl B$ represent arrays of known explanatory variables 
and unknown regression coefficients respectively, and 
$\bs\Sigma(\rho)$ is some simple one-parameter 
model 
that accounts 
for some of the temporal dependence in the data.
In such a model, the reduced rank term 
$\bl S\times \{ \bl U_1 , \cdots ,  \bl U_K  \}$ would
express data patterns not accounted for by
$ \langle \bl X , \bl B \rangle $ or $\bs\Sigma(\rho)$. 
Bayesian inference for parameters in such a model 
could be obtained by adding steps to the 
MCMC  algorithm outlined in this article.

\medskip

Replication code for the results in Sections 4 and 5 
is available at the author's
website:
\href{http://www.stat.washington.edu/~hoff}{\nolinkurl{www.stat.washington.edu/~hoff}}.
This research was supported by NI-CHD grant R01HD067509.

\appendix

\section{Proofs}

\begin{proof}[Proof of Proposition \ref{prop:invdr}]
Suppose a model $\{ p(\bl y|\theta) : \theta\in \Theta\}$ is invariant
under a group $\mathcal G$  that acts properly on $\mathcal Y$ and 
for which
the induced group $\bar {\mathcal G}$ over $\Theta$ is transitive. 
By Theorem 6.5 of \citet{eaton_1989},
a minimum risk equivariant  decision rule
under an invariant loss $L(d,\theta) $ is then given by the minimizer in $d$
of
\begin{equation}
H(d,\bl y) = \int_{\bar {\mathcal G}} 
   L(d,\bar g \theta_0) p(\bl y|\bar g \theta_0) \,  \mu ( d \bar g ) \, , 
\label{eq:eqrisk} 
\end{equation}
where $\theta_0$ is an arbitrary point in $\Theta$ and
 $\mu$ is the right invariant Haar measure on
 $\bar {\mathcal G}$. 
Since
$\bar{\mathcal G}$ is transitive and  the
integrand depends on
$\bar g$ only through $\bar g \theta$,
a change of variables
allows us to
re-express  (\ref{eq:eqrisk}) as
\begin{equation}
H(d,\bl y) = \int_{\Theta} 
   L(d, \theta) p(\bl y| \theta) \,  \pi ( d \theta ) \, ,
\label{eq:bayesrisk}
\end{equation}
where $\pi$ is the measure on $\Theta$ induced by $\mu$ via 
$\pi(A) = \mu ( \{ \bar g : \bar g \theta_0 \in A\} )$. 
As $H(d,\bl y)$ is proportional to the posterior risk under prior
$\pi$, the minimum risk equivariant estimator is equivalent to
the Bayes solution under the (potentially improper) prior $\pi$.

The Tucker decomposition model with known core array $\bl S$ 
is invariant under transformations of the form 
$g_{aW}: \bl y \rightarrow a \bl W \bl y$ and 
$\bar g_{aW} : (\sigma, \bl U) \rightarrow 
   (a \sigma , \bl W \bl U)$ for $a>0$ and $\bl W \in \mathcal W$. 
The set of such transformations $\bar g$ forms a  group $\bar {\mathcal G}$  
with composition as the group action,  
so that $\bar g_{a_1W_1} \bar g_{a_2 W_2} = \bar g_{a_1a_2 W_1W_2}$, 
and note that the elements of the group are uniquely indexed by 
$(a,\bl W) \in \mathbb R^+ \times \mathcal W$.  
Although the group $\mathcal G$ does not act properly  on $\mathbb R^n$, it 
does act properly on $\mathbb R^n \setminus \{ \bs 0\}$, and so 
the above results apply on this reduced sample space that has 
probability one under the model (see \citet[section 6.3]{eaton_1989}). 
It is straightforward to show that the group $\bar{ \mathcal G}$ is transitive over the parameter space:  Given $\theta_1 = (\sigma_1 ,\bl U_1 )$ and  $\theta_2=( \sigma_2 ,\bl U_2  )$, then
$\theta_2 = \bar g \theta_1 $ for  the $\bar g$  given by
$a = \sigma_2/\sigma_1$ and $\bl W = [ \bl U_2 \bl U_2^\perp ] 
   [ \bl U_1 \bl U_1^\perp ]$.

It is first shown that a right invariant Haar measure over this group is 
given by the product of a measure $\mu_1$ over $\mathbb R^+$  having density 
$h(a) \propto 1/a$ with respect to Lebesgue measure, 
and the probability measure  $\mu_2$ over $\mathcal W$ 
induced 
by letting $\bl W \stackrel{d}{=} \bl W_K \otimes\cdots\otimes \bl W_1$, 
where each $\bl W_k$ has the invariant (uniform) probability measure over 
$\mathcal O_{n_k}$, independently for each $k=1,\ldots, K$. 
To see this, let $f$ be any measurable function of $(a,\bl W)$.  
For any $(b, \bl X ) \in  \mathbb R^+ \times \mathcal W$, 
\begin{align*}
\int f( ab, \bl W \bl X ) h(a) \, da \times \mu_2(d\bl W)  &= 
 \int f(\tilde a , \bl W \bl X ) h(\tilde a/b )  \, 
   \left | \frac{ da }{d\tilde a} \right | d\tilde a   \times 
  \mu_2(d {\bl W}), 
\end{align*} 
by the change of variables from $a$ to $\tilde a=ab$. 
Now  $da/d\tilde a  =1/b$ and 
$h(\tilde a/b) = bh(\tilde a)$, and so 
\begin{align*}
\int f( ab, \bl W \bl X ) h(a) \, da \times \mu_2(d\bl W)  &= 
 \int f(\tilde a , {\bl W\bl X} ) h(\tilde a)   \, 
    d\tilde a   \times 
   \mu_2(d {\bl W})     \\
&= \int f(a , {\bl W\bl X} ) h(a)   \, 
    d a   \times 
   \mu_2(d {\bl W}). 
\end{align*}  
Finally, if under $\mu_2$ the $\bl W_k$'s are independent and each
$\bl W_k$ 
has the invariant distribution over 
$\mathcal O_{n_k}$, then  
$\bl W_k \stackrel{d}{=} \bl W_k \bl X_k$ and  
$\bl W \stackrel{d}{=} \bl W  \bl X $. This gives
\begin{align*}
\int f( ab, \bl W \bl X ) h(a) \, da \times \mu_2(d\bl W)  &= 
 \int f( a, \bl W ) h(a) \, da \times \mu_2(d\bl W) 
\end{align*}
for all $(b,\bl X ) \in \mathbb R^+ \times \mathcal W$, thereby showing 
that the measure $\mu= \mu_1 \times \mu_2$ described above is the right 
invariant Haar measure over the set 
$\mathbb R^+ \times \mathcal W$ that indexes $\bar {\mathcal G}$. 

Following \citet[page 86]{eaton_1989}, 
the measure $\mu$ over values of $(a, \bl W)$ induces a
measure $\pi$ over $(\sigma ,\bl U)$, allowing
(\ref{eq:eqrisk}) to be re-expressed as 
(\ref{eq:bayesrisk}).  
The induced measure $\pi$ 
is given by
\begin{align*}
\pi( \{ (\sigma, \bl U): (\sigma, \bl U)  \in  A\times  B \})  &= 
\mu( \{ (a,\bl W) :  a\sigma_0 \in A , \bl W \bl U_0 \in B\})   \\
&=\mu_1( \{a : a\sigma_0 \in A \} ) \times  \mu_2(\{\bl W: \bl W\bl U_0 \in B\})
\end{align*}
for sets  $A\subset \mathbb R^+$ and $B\subset \mathcal U$ and 
arbitrary $(\sigma_0,\bl U_0) \in \mathbb R^+ \times \mathcal U$.
Letting $\sigma_0=1$, one sees that 
$\pi( \{ (\sigma, \bl U)  : \sigma \in A\}) = \mu_1(A)$. 
As for the distribution of $\bl U$ under $\pi$, 
let $\bl U_0 =  \bl J_{K} \otimes \cdots \otimes \bl J_1$ 
where $\bl J_k =  [ \bl I_{r_k\times r_k} \  \bl 0_{r_k \times n_k }]^T$. 
The distribution of $\bl U$ is therefore the same as that of 
$\bl W_K \bl J_K \otimes  \cdots \otimes \bl W_1 \bl J_1$, 
were for each $k$, $\bl W_k$ is uniform on $\mathcal O_{n_k}$. 
However, $\bl W_k \bl J_k$ is simply the $n_k\times r_k$ orthonormal matrix 
made up of the first $r_k$ columns of $\bl W_k$,  which has the 
uniform (invariant) distribution on $\mathcal V_{r_k,n_k}$. 
\end{proof}

\begin{proof}[Proof of Proposition \ref{prop:eqinf}] 
Letting $\theta=(\sigma, \bl U) $, 
the probability $\Pr(\theta \in \bar g A |g\bl y )$  is  
given by
\begin{equation} 
\Pr(\theta \in \bar g A |g\bl y) = 
 \frac{   \int 1(\theta \in \bar g A)  p( g \bl y | \theta, \bl s ) 
     \pi_I( d\theta) \pi_s(d\bl s)   }
      {   \int  p( g \bl y | \theta, \bl s ) 
    \pi_I( d\theta)  \pi_s(d\bl s)  }. 
\label{eq:pf2}
\end{equation}
Now  with $g\bl y = a \bl W \bl y$, one has
\begin{align*} 
p(g \bl y | \theta, \bl s) & =
 (2 \pi)^{-n/2} \sigma^{-n} \exp\{  -[  a^2 \bl y^T\bl y - 
2a \sigma \bl y^T \bl W^T\bl U \bl s +\sigma^2 \bl s^T\bl s ]/[ 2\sigma^2 ]\} \\
&= a^{-n} (2 \pi)^{-n/2} (\sigma/a)^{-n} \exp\{ 
    -[ \bl y^T\bl y - 2 (\sigma/a) \bl y^T \bl W^T \bl U \bl s + 
    (\sigma/a)^2 \bl s^T\bl s ]/[2(\sigma/a)^2 ] \}   \\
&=  a^{-n} p(\bl y | \bar g^{-1} \theta,\bl s) . 
\end{align*}
This constant  $a^{-n}$ appears in both the numerator and 
denominator of (\ref{eq:pf2}),  and so
\[ 
\Pr(\theta \in \bar g A |g\bl y) = 
\frac{\int 1(\bar g^{-1} \theta \in A)  p( \bl y | \bar g^{-1}\theta, \bl s ) 
    \pi_I( d\theta) \pi_s(d\bl s)  }
      {   \int  p(  \bl y |\bar g^{-1} \theta, \bl s ) 
    \pi_I( d\theta)  \pi_s(d\bl s)  } . 
\]
As $\pi_I$ was derived from the right invariant Haar measure over the 
transformations $\bar g$,  one has
\begin{align*}
\Pr(\theta \in \bar g A |g\bl y ) &=
\frac{ \Delta(\bar g^{-1})   \int 1(\theta \in A)  p( \bl y | \theta, \bl s ) 
    \pi_I( d\theta) \pi_s(d\bl s)   }
      {  \Delta(\bar g^{-1})  \int  p(  \bl y | \theta, \bl s ) 
    \pi_I( d\theta)  \pi_s(d\bl s)  }  \\
 &=  \frac{    \int 1(\theta \in A)  p( \bl y | \theta, \bl s ) 
    \pi_I( d\theta) \pi_s(d\bl s)   }
      {    \int  p(  \bl y | \theta, \bl s ) 
    \pi_I( d\theta)  \pi_s(d\bl s)  } 
=  \Pr(\theta \in A | \bl y), 
\end{align*}
where $\Delta$ is the Haar modulus. 
\end{proof}

\section{Description of GDELT data}
A full description of the GDELT project and data  can be found at
\href{http://gdelt.utdallas.edu/}{gdelt.utdallas.edu}. The data 
analyzed in this article were obtained from the historical backfiles 
of the 2012 data available 
at \href{http://gdelt.utdallas.edu/data/backfiles}{gdelt.utdallas.edu/data/backfiles}.
Attention was restricted to 
events involving governmental agencies 
of pairs of countries (their governments, militaries, police, judiciaries
or intelligence agencies). 

Each event was categorized as belonging to one of the following twenty CAMEO 
action types
\citep{schrodt_omur_gerner_hermrick_2008}:
 make public statement;
 appeal;
 express intent to cooperate;
 consult;
 engage in diplomatic cooperation;
 engage in material cooperation;
 provide aid;
 yield;
 investigate;
 demand;
 disapprove;
 reject;
 threaten;
 protest;
 exhibit force posture;
 reduce relations;
 coerce;
 assault;
 fight;
use unconventional mass violence.
In the political science literature it is standard to 
categorize the first nine of these as  ``positive'' and the last 
eleven as  ``negative''  
\citep{arva_2013}. 

For each ordered pair of countries, action type and week of the year, 
the number of 
days within the week  in which the type of relation occurred was recorded. 
This was done to reduce instances in which a single event was recorded 
multiple times in the dataset. 
The number of events in which each country participated was computed, 
from which the thirty  most active countries were identified. 
These included 
Afghanistan (AFG), 
Armenia (ARM), 
Australia (AUS),
Azerbaijan (AZE), 
Canada (CAN), 
China (CHN)
Germany (DEU), 
Egypt (EGY), France (FRA),  Great Britain (GBR), Indonesia (IDN), 
India (IND),
Iran (IRN),  Iraq(IRQ), Israel (ISR), Italy (ITA), Japan (JPN), 
South Korea (KOR), Libya (Libya), Nigeria (NGA), New Zealand (NZL), 
Pakistan (PAK), Philippines (PHL), Palestinian Occupied Territories (PSE), 
Russia (RUS), Saudi Arabia (SAU), Syria (SYR), Turkey (TUR), 
Ukraine (UKR) and the United States (USA).

\bibliographystyle{chicago}
\bibliography{/Users/hoff/Dropbox/SharedFiles/refs}

\end{document}